\documentclass[a4paper,12pt]{amsproc}

\usepackage{amsmath,amssymb,amsthm}
\usepackage{cite}
\usepackage[mathscr]{eucal}

\theoremstyle{plain}
\newtheorem{theorem}{Theorem}
\newtheorem{lemma}{Lemma}
\theoremstyle{definition}
\newtheorem{example}{Example}
\newtheorem{condition}{Condition}
\newtheorem*{remark}{Remark}

\newcommand{\1}[1]{\overset1{#1}}

\newcommand{\bn}{\mathbf{n}}
\def\RR{\mathbf{R}}

\def\lra{\longrightarrow}

\def\a{\alpha}
\def\be{\beta}
\def\ga{\gamma}
\def\la{\lambda}
\def\e{\varepsilon}
\def\ph{\varphi}
\def\Ga{\Gamma}
\def\La{\Lambda}

\DeclareMathOperator\im{Im}

\DeclareMathOperator\supp{supp}
\DeclareMathOperator\const{const}
\DeclareMathOperator\tr{tr}
\DeclareMathOperator{\ind}{ind}

\def\cH{{\mathcal{H}}}
\def\cI{{\mathcal{I}}}
\def\cJ{{\mathcal{J}}}
\def\cP{{\mathcal{P}}}

\def\wt{\widetilde}
\def\wh{\widehat}
\def\pa{\partial}

\newcommand\pd[2]{\frac{\pa #1}{\pa #2}}
\newcommand\od[2]{\frac{d #1}{d #2}}

\let\rom\textup

\newcommand{\BL}{\biggl}
\newcommand{\BR}{\biggr}
\newcommand{\Bl}{\Bigl}
\newcommand{\Br}{\Bigr}
\newcommand{\bl}{\bigl}
\newcommand{\br}{\bigr}

\newcommand{\abs}[1]{\lvert#1\rvert}

\begin{document}

\author{S.~Yu.~Dobrokhotov}
\address{A.~Ishlinsky Institute for Problems in Mechanics,
Moscow; Moscow Institute of Physics and Technology, Dolgoprudny,
Moscow District}
\email{dobr@ipmnet.ru}
\thanks{Supported by RFBR (grant~11-01-00973-a) and by the
Archimedes Center for Modeling, Analysis \& Computation (ACMAC),
Crete, Greece (grant FP7-REGPDT-2009-1). S.~Yu.~D. is grateful to
the staff of the Institute for Molecules and Materials of the
Radboud University Nijmegen, The Netherlands. S.~Yu.~D. and
V.~E.~N. are grateful to the staff of ACMAC and the Department of
Applied Mathematics, University of Crete for support and kind
hospitality.}

\author{G.~Makrakis}
\address{Department of Applied
Mathematics, University of Crete; Institute of Applied \&
Computational Mathematics, Foundation for Research and
Technology-Hellas, Heraklion, Crete, Greece}
\email{makrakg@iacm.forth.gr}

\author{V.~E.~Nazaikinskii}
\address{A.~Ishlinsky Institute for Problems in Mechanics,
Moscow; Moscow Institute of Physics and Technology, Dolgoprudny,
Moscow District}
\email{nazay@ipmnet.ru}

\author{T.~Ya.~Tudorovskii}
\address{Radboud
University Nijmegen, Institute for Molecules and Materials, NL-6525
AJ Nijmegen, The Netherlands}
\email{t.tudorovskiy@science.ru.nl}

\title[New formulas for Maslov's canonical operator]%
{New formulas for Maslov's canonical operator\\
in a neighborhood of focal points and caustics \\
in 2D semiclassical asymptotics}

\maketitle

\section{INTRODUCTION}

Maslov's canonical operator \cite{Mas1} (see also
\cite{MaFe1,MSS1,Mas6,MaNa2,BelDob92}) is used when constructing
short-wave (high-frequency, or rapidly oscillating) asymptotic
solutions for a broad class of differential equations with real
characteristics. The asymptotics given by the canonical operator
are a far-reaching generalization of ray expansions in problems of
optics, electrodynamics, etc.\ and of WKB asymptotics for equations
of quantum mechanics. These asymptotics are based on some solutions
of the equations of classical (Hamiltonian) mechanics and in a
sense permit automatically and globally writing out solutions of
equations of quantum and wave mechanics taking into account the
focal points and caustics occurring in the problem. The
oscillations are usually characterized by a large positive
parameter~$k$ in problems of optics and by a small positive
parameter~$h$ in problems of quantum mechanics, and the asymptotics
should be constructed as $k\to +\infty$ and $h\to +0$,
respectively; in the present paper, we use the parameter~$h$. The
construction of Maslov's canonical operator is based on a
fundamental geometric object known as a \textit{Lagrangian
manifold}.  While the original differential equation lives in an
$n$-dimensional configuration space~$\mathbb{R}^n_x$ with
coordinates $x=(x_1,\ldots,x_n)$, the Lagrangian manifold, which we
denote by~$\Lambda^n$ is a smooth $n$-dimensional manifold in the
phase space $\mathbb{R}^{2n}_{px}$ with coordinates $(p,x)$,
$p=(p_1,\ldots,p_n)$, $x=(x_1,\ldots,x_n)$. In the one-dimensional
case, where Lagrangian manifolds~$\Lambda^1$ are curves on the
phase plane $(p,x)$, one can manage without using these manifolds
when constructing asymptotic solutions by the WKB method, but even
in this case they prove to be rather useful. Let
$\alpha=(\alpha_1,\ldots,\alpha_n)$ be coordinates on~$\Lambda^n$;
then one can specify~$\Lambda^n$ by the formulas
$\Lambda=\{p=P(\alpha),x=X(\alpha)\}$.  In physical problems, as a
rule, the coordinates~$\alpha$ vary on the product of the
$k$-dimensional Euclidean space~$\mathbb{R}^k$ and the
$n-k$-dimensional torus~$\mathbb{T}^{n-k}$ or the $n-k$-dimensional
sphere~$\mathbb{S}^{n-k}$. When constructing a solution
on~$\Lambda^n$, one should specify an amplitude $A=A(\a)$, a
measure (volume element)~$d \mu$ and a distinguished
point~$\alpha_0$ on~$\Lambda^n$ (the so-called \textit{central
point}\footnote{A different choice of the central point is
equivalent to multiplying the canonical operator by a constant
phase factor $e^{i\theta}$.}).

Maslov's canonical operator~$K_{\La^n}^h$ takes a function~$A$
on~$\Lambda^n$ to a function~$u(x,h)$of $x\in \mathbb{R}^n_{x}$; we
denote this correspondence by
\begin{equation}\label{KO}
    u(x,h)=[K^h_{\Lambda^n}A](x).
\end{equation}
An important role in the definition and properties of the canonical
operator is played by the \textit{Lagrangian
singularities}~$\Sigma$, which are defined as the set of zeros of
the Jacobian $\mathcal{J}=\det \frac{\pa X}{\pa \alpha}$
corresponding to the projection of~$\Lambda^n$ onto the
configuration space~$\mathbb{R}^n$. The points of~$\Sigma$ are said
to be \textit{focal},and the projection of~$\Sigma$
onto~$\mathbb{R}^n$ specifies the \textit{caustics} of the wave
field~$u(x,h)$. The definition of~$u(x,h)$ involves some additional
objects and is nonunique, but this nonuniqueness gives only small
changes in~$u(x,h)$ as~$h\to 0$ unimportant from the viewpoint of
physical applications. Finally, note that Maslov's canonical
operator is an object of function theory on its own, even though
its main applications are related to partial differential
equations.

One main idea underlying the canonical operator is to pass from the
original differential equation in the space~$\mathbb{R}^n_x$ to a
simpler induced equation on~$\Lambda^n$. The manifold~$\Lambda^n$
is not universal even for a fixed differential equation; it depends
on the problem considered for that equation. The solution~$A$ of
the reduced equation on~$\Lambda^n$, known as the
\textit{amplitude}, depends on the problem as well. It is very
important that $A$~is a smooth function on~$\Lambda^n$ even in the
vicinity if the Lagrangian singularities, in contracts to the
amplitudes in the traditional ray (WKB) expansions. For many types
of problems (and for various original differential equations),
there exist recipes or algorithms for constructing the
corresponding manifolds and amplitudes. Once $\Lambda^n$~and~$A$
have been obtained, the solution~$u(x)$ of the original problem for
the corresponding differential equation can be reconstructed by
formula~\eqref{KO}. In other words, given~$\Lambda^n$ and~$A$,
$[K^k_{\Lambda^n}A](x)$ is the \textit{answer to the problem}, and
this answer automatically includes objects and operations of ray
expansions such that the behavior in caustic domains, passage
across the caustics, matching of various asymptotic
representations, etc. Thus, the problem is reduced to the
construction of~$\Lambda^n$ and~$A$ and to the simplification of
the expression $[K^k_{\Lambda^n}A](x)$ in specific cases.

The right-hand side of~\eqref{KO} can only very vaguely be called a
formula; it is rather an algorithm or a set of rules that permit
one to implement~\eqref{KO} in the form of more or less closed-form
expressions containing rapidly oscillating exponentials or
integrals of such exponentials. We point out that, first, that
these formulas are not as a rule the same (universal) for all
values of the variables~$x$; they have different asymptotic
representations in different domains (depending on the problem).
Second, even in one and the same domain these representations can
be defined nonuniquely, and a lucky choice of a representation may
substantially simplify the (local) form of the solution and permit
one to represent it, for example, via well-known special or even
elementary functions. Maslov suggested a universal recipe for
representing the function\footnote{In contrast to the global
function $[K^h_{\Lambda^n}A](x)$, this recipe depends on the
coordinate system chosen in $\mathbb{R}^n_x$.}
$[K^h_{\Lambda^n}A](x)$ in a neighborhood of the caustics on the
basis of the partial Fourier transform (i.e., the Fourier transform
with respect to part of the variables). This recipe applies in the
most general situation, but, for a rather broad class of
interesting problems, one can (more conveniently) use different
representations that are not related to the choice of a partial
Fourier transform. The present paper, which deals with the
two-dimensional case ($n=2$), presents a new integral
representation (Eq.~\eqref{eq:02-10}) of oscillatory solutions for
the case in which the fundamental $1$-form $p\,dx$ nowhere vanishes
on the corresponding Lagrangian manifolds~$\Lambda^n$. We also give
some applications. We point out that our considerations do not
affect the general concept of the construction of Maslov's
canonical operator and the fundamental underlying objects; we only
suggest a more convenient implementation useful in specific
physical problems, for example, those related to the asymptotics of
solutions of the scattering problem, asymptotics of the Green
function, linear hyperbolic systems with variable coefficients
(e.g., the wave equation) with localized initial data (e.g., see
\cite{DShTMZ,DShT}), etc. Note also that our formulas are in a
sense a special case of the general formulas of the theory of
Fourier integral operators~\cite{Hor6}, and our main result is a
specific (constructive) form and an algorithm for the construction
of these formulas, which can in particular be used in combination
with software like \textsl{Mathematica} or \textsl{MatLab}.

The paper is organized as follows. In Sec.~2, we consider an
important example illustrating the idea of a new formula for
Maslov's canonical operator and explaining why it is tempting to
write it out. This new formula~\eqref{eq:02-10} and associated
objects are presented in Sec.~3. Section~4 contains some examples.
The proof of the main theorem and the formulas expressing Maslov's
canonical operator in a neighborhood of the caustics via the Airy
and Pearcy functions are given in Appendices~1 and~2, respectively.

\subsubsection*{Some notation}

All vectors are understood as column vectors. If $\xi$ and $\eta$
are $n$-vectors, then we write $\langle \xi,\eta\rangle$ for the
bilinear form $\langle \xi,\eta\rangle =\sum_{j=1}^n
\xi_j\eta_j=\xi^T\eta$, where the symbol $T$ indicates the
transpose of a matrix. Partial derivatives are denoted by
subscripts; for example, $\Phi_x=\pa\Phi/\pa x$.

\section{LAGRANGIAN MANIFOLD FOR THE BESSEL FUNCTION}
\label{sec:21}

There are only a few types of Lagrangian manifolds arising in
specific physical applications. The simplest examples are
Lagrangian surfaces. Let $S(x)$, $x\in \mathbb{R}^n$, be a smooth
function; then the equation $p=\frac{\pa S}{\pa x}$ specifies a
surface in the phase space $\mathbb{R}^{2n}_{px}$. This surface is
a Lagrangian manifold, and associated with this manifold are
functions $u(x,h)$ of the form $A(x)e^{\frac{iS(x)}{h}}$, known as
\textit{WKB solutions}. For such a function to be an asymptotic
solution of the Helmholtz equation $h^2\triangle u+n^2(x)u=0$, it
is in particular necessary that the function~$S$ satisfy the
Hamilton--Jacobi equation ${\nabla S}^2=n^2(x)$.

Let us present an example of a 2D Lagrangian manifold, which is the
main example for this paper. This manifold corresponds to the
Helmholtz equation with $n^2(x)=1$.  We start from the following
simple problem: construct rapidly oscillating functions $u(x,h)$,
$x=(x_1,x_2)\in \mathbb{R}_x^2$, associated with the
two-dimensional Lagrangian cylinder
\begin{equation}\label{eq:01-1}
 \begin{gathered}
  \Lambda^2=\{(x,p)\colon x=X(\tau,\psi),\;
p=P(\tau,\psi),\;\; \tau\in\mathbb{R},\;
    \psi\in \mathbb{S}^1=\mathbb{R}\!\!\!\!\pmod{2\pi}\},
    \\\text{where}\quad
   X(\tau,\psi)=\tau\bn(\psi),\quad
   P(\tau,\psi)=\bn(\psi),\quad
     \bn(\psi)=(\cos\psi,\sin\psi)^T,
\end{gathered}
\end{equation}
in the four-dimensional phase space $\mathbb{R}_{(x,p)}^4$ with
coordinates $(x,p)=(x_1,x_2,p_1,p_2)$. The functions $(\tau,\ph)$
form a coordinate system on~$\Lambda^2$. This manifold was used
in~\cite{DShTMZ,DShT} in the representation of rapidly decaying
function of the form $f\bl(\frac{x}{h}\br)$, $h\ll1$. The
projection of the manifold $\Lambda^2$ onto the configuration space
(plane) $\mathbb{R}^2_x$ is a two-sheeted covering with a
singularity at $x=0$. One can readily compute the Jacobian
$\mathcal{J}=\det(X_\tau,X_\psi)=\tau$. Its zeros
(\textit{Lagrangian singularities}) $\Sigma$ are determined by the
equation $\tau=0$, which specifies a circle on~$\Lambda^2$. The
projection of this circle into $\mathbb{R}^2_x$ is the
\textit{degenerate caustic} consisting of the single point $x=0$,
and hence the manifold $\Lambda^2$ is not in \textit{general
position} in the sense of catastrophe theory. Nevertheless,
manifolds of this type and their generalizations play an important
role in physical applications.

Recall that the points of~$\Lambda^2$ where $\mathcal{J}\neq0$ are
said to be \textit{regular}, in contrast to the \textit{singular}
(\textit{focal}) points, where  $\mathcal{J}=0$. The topological
characteristic known as the \textit{Maslov index} plays an
important role in asymptotic formulas. The Maslov index is defined
in our example is as follows.  The circle $\tau=0$ divides
$\Lambda^2$ into two parts  $\Omega_\pm$ consisting of regular
points with $\tau>0$ and $\tau<0$, respectively, with the same
Maslov index $m_+$ for all points in $\Omega_+$ and the same Maslov
index $m_-$ for all points in $\Omega_-$. Fix a point on
$\Lambda^2$ with coordinates $\psi_0,\tau_0=+0$ and define the
Maslov index $m(\psi_0,\tau_0)=0$; then $m_+=0$, and one can
readily prove that $m_-=1$ (see Example~\ref{3.1} below and, e.g.,
\cite{DShTMZ,DShT}).

Let us write out the expressions provided for the rapidly
oscillating functions by the standard construction of Maslov's
canonical operator $[K_{\La^2}^h a](x)$, acting on a smooth
function $a(\tau,\psi)$ on $\La^2$. Outside the caustic $x=0$, we
have the WKB function
\begin{equation}\label{WKB1}
\begin{split}
 u(x,h)=[K_{\La^2}^ha](x)
       &\equiv \sum_{\pm}
        \frac{1}{\sqrt{|J(\tau_\pm(x))|}}
        \bl(e^{\frac{i}{h}\tau_\pm(x)}e^{-\frac{i\pi m_\pm}{2}}a(\tau_\pm(x),\psi_\pm(x)\br)
 \\
       &=\frac{e^{-i\frac{\pi}{4}}}{\sqrt{|x|}}
       \sum_{\pm}(e^{\pm \frac{i}{h} (|x|+\pi/4)}a(\pm|x|,\psi_\pm(x)).
\end{split}
\end{equation}
Here $\tau_\pm(x)=\pm|x|$, the functions $\psi_+(x)=\varphi(x)$ and
$\psi_-(x)=\varphi(x)+\pi$ are the solutions of the equations
\begin{equation}\label{WKB2}
    \tau \cos\psi=x_1,\qquad \tau\sin\psi=x_2
\end{equation}
for positive (the $+$ sign) and negative (the $-$ sign) $\tau$,
respectively, and $\varphi(x)$ is the polar angle of the vector
$x$. To construct the function $u(x,k)=[K_{\La^2}^h a](x)$
globally, including a neighborhood of the caustic $x=0$, one covers
a neighborhood of the preimage of $x=0$ on $\Lambda^2$ by the
\textit{canonical charts}
\begin{align*}
 \Omega_{1}&=\{\psi\in
(-3\pi/8, 3\pi/8)\},&\Omega_{3}&=\{\psi\in
 (5\pi/8, 11\pi/8)\}\quad \text{with coordinates $(x_1,p_2)$} ,\\
 \Omega_{2}&=\{\psi\in
 (\pi/8, 7\pi/8)\},&
 \Omega_{4}&=\{\psi\in
(9\pi/8, 15\pi/8)\}\quad \text{with coordinates $(p_1,x_2)$}.
\end{align*}
Let $1=\sum_{j=1}^4e_j(\psi)$ be a partition of unity on $\La^2$
subordinate to the cover of $\La^2$ by these neighborhoods. Then,
up to terms of lower order as $h\to +0$, Maslov's canonical
operator $K_{\La^2}^h$ applied to a function $a$ on $\La^2$ is
given by the formula
\begin{multline}\label{eq:01-4}
  u(x,h)=[K_{\La^2}^h a](x)\equiv
\sum_{j=1,3}\BL(\frac{i}{2\pi h}\BR)^{1/2}
\int_{-\infty}^{+\infty}\frac{e^{\frac{i}{h}(\tau \cos^2
\psi+x_2p_2)}a(\tau,\psi)e_{j}(\psi)} {\abs{\cos\psi}}
\bigg\vert{}_{\substack{\psi=\psi_{j}(x_1,p_2)\\
\tau=\tau_{j}(x_1,p_2)}} d p_2
 \\
 + \sum_{j=2,4}\BL(\frac{i}{2\pi h}\BR)^{1/2}
\int_{-\infty}^{+\infty} \frac{e^{\frac{i}{h}(\tau \sin^2 \psi+
x_1p_1)}a(\tau,\psi)e_{j}(\psi)} {|\sin\psi|}
\bigg\vert{}_{\substack{\psi=\psi_{j}(p_1,x_2)\\
\tau=\tau_{j}(p_1,x_2)}}
 d p_1,
\end{multline}
where $i^{1/2}=e^{i\pi/4}$ and the functions $\tau_{j}$ and
$\psi_{j}$ express the global coordinates $(\tau,\psi)$ on $\La^2$
via the coordinates in $\Omega_{j}$ (i.e., via $(x_1,p_2)$ for
$j=1,3$ and via $(p_1,x_2)$ for $j=2,4$). \textit{We point out
that, modulo small correction, \eqref{eq:01-4} is independent of
the choice of the charts $\Omega_{j}$ and the partition of unity
$\{e_j\}$.}

\medskip

Note that formula \eqref{eq:01-4} can be significantly simplified,
especially from the view point of specific applications. Namely,
one can replace the integration over the momenta $p_1$ and $p_2$ by
integration over the angle $\psi$ in each chart $\Omega_{j}$ by
setting $p_1=\cos \psi$ or $p_2=\sin \psi$, depending on whether
$j$ is even or odd. This gives
\begin{align}\label{eq:01-4a}
  u(x,h)&= \BL(\frac{i}{2\pi h}\BR)^{1/2}
        \int_{0}^{2\pi}e^{\frac{i}{h}(x_1 \cos
        \psi+x_2\sin \psi)}A(x,\psi)d \psi,
\qquad\text{where}\\ \label{eq:01-4b}
 A(x,\psi)&=
\sum_{j=1,3}a\BL(\frac{x_1}{\cos\psi},\psi\BR)e_{j}(\psi) +
\sum_{j=2,4}a\BL(\frac{x_2}{\sin\psi},\psi\BR)e_{j}(\psi).
\end{align}
If the function  $a(\tau,\psi)$ is independent of $\tau$, then
$A=a(\psi)$. Moreover, if $a=1$, then the function \eqref{eq:01-4a}
is, up to a multiplicative constant, just the \textit{zero-order
Bessel function}
$u(x,h)=\mathbf{J}_0\Bl(\frac{\sqrt{\smash[b]{x_1^2+x_2^2}}}h\Br)$,
and \eqref{WKB1} is none other than the leading term of its
asymptotics for large values of the argument.

This example shows that definition \eqref{eq:01-4a},
\eqref{eq:01-4b}, based on integration over the angle, is more
constructive and pragmatic than the standard definition of Maslov's
canonical operator based on integration over momenta; in particular
it does not require splitting into four charts in the corresponding
formulas; however, the ``practical'' drawback in
definition~\eqref{eq:01-4b} of the function $A$ is its
``noninvariant'' form with respect to the choice of the charts
$\Omega_{j}$ and the partition of unity $e_j$, although we again
point out that the final result is invariant modulo a small
correction. \textit{The main goal of this paper and the formulas
constructed below is to provide a representation of Maslov's
canonical operator in a neighborhood of Lagrangian singularities
based on the integration over ``angle variables'' similar to $\psi$
and directly involving the function $a$ on the corresponding
Lagrangian manifold without a partition of unity etc.}

\medskip

Let us show how one can naturally construct a function of the form
\eqref{eq:01-4a} in our example without using the standard
representation \eqref{eq:01-4}.

\begin{small}
Just as in the general case in Sec.~\ref{sec:22} below, we use a
specific form of the universal construction of the theory of
Fourier integral operators~\cite{Hor6}, which (being restated for
the case of asymptotics with respect to the small parameter~$h$)
say that, to construct the rapidly oscillating functions
corresponding to a given Lagrangian manifold~$\La$, one should find
a real-valued \textit{nondegenerate phase
function}~$\Phi(x,\theta)$ depending on parameters $\theta\in\RR^m$
($m\ge0$) and \textit{determining}~$\La$ in the sense that the
differentials $d\Phi_{\theta_j}(x,\theta)$, $j=1,\dotsc,m$, are
linearly independent at the points where $\Phi_\theta(x,\theta)=0$
and one has the representation $\La^2=\{(x,p)\colon \exists
\theta\,\Phi_\theta(x,\theta)=0,\; p=\Phi_x(x,\theta)\}$. Then the
desired rapidly oscillating functions have the form
\begin{equation}\label{RO}
  u(x,h)=\BL(\frac i{2\pi h}\BR)^{m/2}
  \idotsint
  e^{\frac{i}{h}\Phi(x,\theta)}A(x,\theta)\,d\theta_1\dotsm
  d\theta_m
\end{equation}
with some \textit{amplitude}~$A(x,\theta)$, which is a smooth
function compactly supported in~$\theta$. The universality of this
construction is in particular shown by the following theorem.
\begin{theorem}\label{th0}
Let $\Phi_1(x,\theta)$, $\theta\in\RR^{m_1}$, and
$\Phi_2(x,\theta)$, $\theta\in\RR^{m_2}$, be two nondegenerate
phase functions determining the same Lagrangian manifold~$\La$.
Then the corresponding oscillatory integrals of the form~\eqref{RO}
specify one and the same class of rapidly oscillating functions
\rom(and one can explicitly write out the transformation of
amplitudes in the passage from one representation to the
other\rom).
\end{theorem}
In this general form, the theorem is not used in the present paper.
Hence we refer the reader for the proof (and a more detailed
statement concerning the transformation of amplitudes)
to~\cite[Theorem~4 and Corollary~1]{arXiv}. Theorem~\ref{th000}
below, which is a special case of Theorem~\ref{th0}, is proved in
detail in Appendix~\ref{sec:02-06}.
\end{small}
\medskip

Let us return to our example. The construction described above is
local in general, but for the manifold $\La^2$ we can readily find
a global defining function $\Phi(x,\theta)$. Let $x\in\mathbb{R}^2$
and $\psi\in \mathbb{S}^1$. On the straight line
\begin{equation*}
  \ell_\psi=\{y\in \mathbb{R}^2\colon y=X(\tau,\psi),
  \;\tau\in\mathbb{R}\},
\end{equation*}
take the point $x_*=X(\tau_*,\psi)$ nearest to~$x$, where
$\tau_*=\langle x,\bn(\psi)\rangle$ (sice the segment $[x,x_*]$ is
orthogonal to $\ell_\psi$). Using $\psi$ in the role of the
variable~$\theta$, set
\begin{equation}\label{fufu}
  \Phi(x,\psi)=\tau_*\equiv\langle x,\bn(\psi)\rangle.
\end{equation}
This phase function is nondegenerate, because
\begin{equation*}
  \Phi_\psi=\langle x,\bn'(\psi)\rangle, \qquad
 (\Phi_\psi)_x=\bn'(\psi)\equiv(-\sin\psi,\cos\psi)\ne0.
\end{equation*}
Next, the equation $\Phi_\psi=0$ means that $x$~and~$\bn'(\psi)$
are orthogonal; in other words, $x$~is collinear to~$\bn(\psi)$ and
hence lies on~$\ell_\psi$, so that $x=X(\tau_*,\psi)$. Moreover,
$\Phi_x(\tau_*,\psi)=\bn(\tau_*,\psi)=P(\tau_*,\psi)$, so that the
point $(x,\Phi_x(\tau_*,\psi))$ lies on~$\La^2$. It is easily seen
that every point in~$\La^2$ can be obtained in such a way, so that
the phase function~\eqref{fufu} globally defines the
manifold~$\La^2$.

The amplitude~$A(x,\psi)$ can be an arbitrary smooth function, so
that the rapidly oscillating functions associated with the
manifold~$\La^2$ have the form
\begin{equation}\label{RO1}
  u(x,h)=\BL(\frac i{2\pi h}\BR)^{1/2}
  \int_0^{2\pi}e^{\frac{i}{h}\langle x,\bn(\psi)\rangle}
  A(x,\psi)\,d\psi.
\end{equation}
Note, however, that this representation is asymptotically
nonunique: \textit{if one replaces $A(x,\tau)$ by any other smooth
function $A'(x,\tau)$ such that}
\begin{equation}\label{changeA}
 A(x,\psi)=A'(x,\psi) \quad \text {on the set}
 \qquad C_\Phi=\{(x,\psi)\colon \Phi_\psi\equiv\langle
 \mathbf{n}'(\psi),x\rangle=0\},
\end{equation}
\textit{then the integral \eqref{RO1} changes only by~$O(h)$}.

\medskip

\begin{small}
Indeed, if the amplitude is zero at the points where $\Phi_\psi=0$,
then it can be represented in the from of the product
$B(x,\psi)\Phi_\psi(x,\psi)$, and one can show by integrating by
parts in the integral~\eqref{RO1} that $u(x,h)=O(h)$.
\end{small}

\medskip

For example, it follows that the integral~\eqref{RO1} will not
change in the leading term of the asymptotics as $h\to0$ if one
replaces $A(x,\psi)$ by $A(x_*,\psi)$. Note that
\begin{equation*}
    A(x_*,\psi)=A(\tau_*\bn(\psi),\psi)
    =A(\langle x,\bn(\psi)\rangle\bn(\psi),\psi)
    =a(\langle x,\bn(\psi)\rangle,\psi),
\end{equation*}
where $a(\tau,\psi)=A(\tau\bn(\psi),\psi)$, so that the rapidly
oscillating function associated with $\La^2$ can be represented in
the form (cf.~\eqref{eq:01-4a})
\begin{equation}\label{eq:01-2}
  u(x,h)=\BL(\frac i{2\pi h}\BR)^{1/2}\int_0^{2\pi}
  e^{\frac{i}{h}\langle x,\bn(\psi)\rangle}
  a(\langle \mathbf{n}(\psi),x\rangle,\psi)\,d\psi.
\end{equation}
This representation, in contrast to~\eqref{RO1}, is asymptotically
unique.

In particular, if we drop the normalizing factor multiplying the
integral and take the constant $a(x,\psi)=(2\pi)^{-1}$ for the
amplitude, then we obtain the well-known integral representation
\begin{equation}\label{eq:01-3}
  \mathbf{J}_0\BL(\frac{\abs{x}}{h}\BR)=\frac1{2\pi}\int_0^{2\pi}e^{\frac{i}{h}\langle
x,\bn(\psi)\rangle}\,d\psi
\end{equation}
of the zero-order Bessel function as a special case of our
construction.

\begin{remark}
Sometimes it is convenient to choose the amplitude in a form
different from that in~\eqref{eq:01-2}. For example, using a change
of the form~\eqref{changeA}, one can transform~\eqref{eq:01-2} as
follows: Set $A(x,\psi)=a(|x|,\psi)$ if the function $a(\tau,\psi)$
is even in~$\tau$ and $A(x,\psi)=\langle \mathbf{n}(\psi),x\rangle
a(|x|,\psi)/|x|$ if $a(\tau,\psi)$ is odd in~$\tau$, and in the
general case split $a(\tau,\psi)$ into the odd and even parts and
set
\begin{equation*}
A=\frac{a(|x|,\psi)+a(-|x|,\psi)}{2}+\langle
\mathbf{n}(\psi),x\rangle\frac{a(|x|,\psi)-a(-|x|,\psi)}{2|x|}.
\end{equation*}
Then (again omitting the normalizing factor multiplying the
integral) modulo $O(h)$ we obtain
\begin{equation}\label{eq:01-2b}
\begin{split}
u(x,h)&=\int_0^{2\pi}e^{\frac{i}{h}\langle x,\bn(\psi)\rangle}
\Bigl(\frac{a(|x|,\psi)+a(-|x|,\psi)}{2}
\\&\qquad\qquad{}+\langle
\mathbf{n}(\psi),x\rangle\frac{a(|x|,\psi)-a(-|x|,\psi)}{2|x|}\Bigr)\,d\psi
\\&=
\int_0^{2\pi}e^{\frac{i}{h}\langle
x,\bn(\psi)\rangle}\frac{a(|x|,\psi)+a(-|x|,\psi)}{2}\,d\psi
\\&\qquad\qquad{}-i\frac{\pa}{\pa
k}\BL(\int_0^{2\pi}e^{ik\langle x,\bn(\psi)\rangle}
\frac{a(|x|,\psi)-a(-|x|,\psi)}{2|x|}\,d\psi\BR)\bigg|_{k=1/h}.
\end{split}
\end{equation}
In the general case, this representation looks more complicated
than~\eqref{eq:01-2}, but if, say, $a$ is independent of the
angle~$\psi$ altogether, $a=a(\tau)$, then it readily leads to
significant simplifications; we obtain
\begin{align}\nonumber
  u(x,h)&=
\pi \bl({a(|x|)+a(-|x|)}\br)\mathbf{J}_0\BL(\frac{\abs{x}}{h}\BR)-
i\pi\bl({a(|x|)-a(-|x|)}\br)
\mathbf{J}_0'\BL(\frac{\abs{x}}{h}\BR)
 \\\label{eq:01-2c}
 &=\pi
\bl({a(|x|)+a(-|x|)}\br)\mathbf{J}_0\BL(\frac{\abs{x}}{h}\BR)+
i\pi\bl({a(|x|)-a(-|x|)}\br) \mathbf{J}_1\BL(\frac{\abs{x}}{h}\BR)
,
\end{align}
where $\mathbf{J}_1(y)$ is the first-order Bessel function.
\end{remark}
\begin{remark}
When writing out formula~\eqref{eq:01-2}, we have ignored the
Maslov index in the singular chart by omitting the factor $e^{-i\pi
m_s/2}$. In fact, one can show that $m_s=0$ in this example
provided that one chooses a nonsingular initial point on~$\La^2$
with coordinates $(\tau,\psi)$, $\tau>0$.
\end{remark}
\begin{remark}
Let us compare the integrals~\eqref{eq:01-2} and~\eqref{eq:01-2b}
for the case in which $a=\tau^2$. Formula~\eqref{eq:01-2b} gives
$2\pi|x|^2\mathbf{J}_0\bl(\frac{\abs{x}}{h}\br)$, while formula
\eqref{eq:01-2b} gives
\begin{align*}
\int_0^{2\pi}e^{\frac{i}{h}\langle x,\bn(\psi)\rangle} (\langle
\mathbf{n}(\psi),x\rangle)^2\,d\psi &= -2\pi\frac{\pa^2}{\pa
(1/h)^2}\mathbf{J}_0\BL(\frac{\abs{x}}{h}\BR)
\\&=-2\pi|x|^2\mathbf{J}_0''\BL(\frac{\abs{x}}{h}\BR)=
2\pi|x|^2\mathbf{J}_0\BL(\frac{\abs{x}}{h}\BR)-2\pi
{|x|}{h}\mathbf{J}_1\BL(\frac{\abs{x}}{h}\BR).
\end{align*}
Thus, the difference between the two representations is $-2\pi
{|x|}{h}\mathbf{J}_1\bl(\frac{\abs{x}}{h}\br)=O(h)$.
\end{remark}

\section{GENERAL FORMULAS FOR CANONICAL OPERATOR
IN 2D~CASE}\label{sec:22}

We proceed to the description of general formulas and constructions
in the two-di\-men\-sional case. Let $\La^2$~be a Lagrangian
manifold in the four-dimensional phase space~$\mathbb{R}_{px}^4$.
The functions specifying the manifold~$\La^2$ (i.e., the embedding
$\La\subset {\mathbb{R}}_{px}^4$) will be denoted by $x=X(\a)$,
$p=P(\a)$, where $\alpha=(\alpha_1,\alpha_2)$ are coordinates on
$\La^2$. To simplify the notation we denote points of~$\La^2$
by~$\a$ as well.

\subsection{Eikonal \rom(action\rom)}\label{sec:02-01}

Since $\La^2$ is Lagrangian, it follows that the Pfaff equation
\begin{equation}\label{eq:02-eiko}
    d\tau(\a)=P(\a)\,dX(\a)\equiv
    P_1(\a)\,dX_1(\a)+P_2(\a)\,dX_2(\a)
\end{equation}
is locally solvable on $\La$. (More precisely, it is solvable in an
arbitrary simply connected domain $U\subset \La^2$). A real
solution $\tau(\a)$ of Eq.~\eqref{eq:02-eiko} is called an
\textit{eikonal} (or \textit{action}) in $U$; if $U$ is connected,
then the eikonal is defined up to an additive constant.

\subsection{Eikonal coordinates}\label{sec:02-02}

From now on, we assume that $\La^2$ satisfies the following
\begin{condition}\label{co:02-1}
The form $P(\a)\,dX(\a)$ is nonzero for every $\a\in\La^2$.
\end{condition}
Thus, if $\tau$ is an eikonal in a neighborhood $U$ of some point
of $\La^2$, then $d\tau\ne0$, and hence (provided that $U$ is
sufficiently small) we can supplement $\tau$ with another function
$\psi$ such that $(\tau,\psi)$ is a coordinate system in $U$. A
coordinate system of this kind will be called an \textit{eikonal
coordinate system}. We see that there exists an eikonal coordinate
system in a neighborhood of an arbitrary point $\a\in\La^2$. The
expressions of the functions $(X(\a),P(\a))$ via eikonal
coordinates will be denoted by $(X(\tau,\psi),P(\tau,\psi))$
(rather than the technically correct
$(X(\a(\tau,\psi)),P(\a(\tau,\psi)))$) or even simply by $(X,P)$
with the arguments omitted. The same notation will  be used for
other functions on $\La^2$.

\begin{lemma}\label{pr:02-01}
In eikonal coordinates, one has the relations
\begin{equation}\label{eq:02-01}
    \langle P,X_\tau\rangle=1,\qquad
    \langle P,X_\psi\rangle=0,\qquad
    \langle P_\psi,X_\tau\rangle=\langle P_\tau,X_\psi\rangle.
\end{equation}
\end{lemma}

\subsection{Measure and Jacobians}\label{sec:02-03}

As was already noted, to construct the canonical operator, one
needs not only the Lagrangian manifold $\La^2$ but also some real
measure on this manifold. We assume that the measure is represented
by a volume form \footnote{Thus, we assume that $\La^2$ is
orientable; the theory may pretty well be constructed without this
assumption, which we only make to simplify the exposition.} $d\mu$.
In eikonal coordinates $(\tau,\psi)$ on $\La^2$, one has
$d\mu=\mu\,d\tau\wedge d\psi$, where $\mu\equiv\mu(\tau,\psi)$ is a
smooth nonvanishing function called the \textit{density} of the
measure $d\mu$ in the coordinates $(\tau,\psi)$. Note that in many
physical problems the coordinate~$\tau$ is the so-called ``proper
time'' and $\mu=1$. We introduce the Jacobians
\begin{align}\label{eq:02-Je0}
\mathcal{J}&=\frac{D(X)}{D\mu}=\frac{1}{\mu}\det\frac{\pa (X_1,X_2)
}{\pa(\tau,\psi)}=\frac{1}{\mu}\det(X_\tau,X_\psi),\quad \widetilde{\mathcal{J}}=
\mu\det (P,P_\psi)
\\ \label{eq:02-Je}
\mathcal{J}^{\e}&=\frac{D(X-i\e P)}{D\mu}
  \overset{\text{def}}{=}\frac{1}{\mu}\det\frac{\pa (X_1-i\e P_1,X_2-i\e P_2) }{\pa(\tau,\psi)},
  \qquad\e\in[0,1].
\end{align}
In contrast to $\mu$, the Jacobians \eqref{eq:02-Je0} and
\eqref{eq:02-Je} are \textit{independent} of the specific choice of
eikonal coordinates and hence are well defined globally on $\La^2$.
\begin{lemma}\label{pr:02-02}
One has
\begin{equation}\label{eq:02-02}
  \abs{\cJ}=\frac{|{X_\psi}|}{{\mu}|{P}|};\qquad
   \cJ^{\e}(\a)\ne0 \quad\text{for any $\a\in\La$ and $\e>0$}.
\end{equation}
\end{lemma}
\begin{proof}
(i) (cf.~\cite{DShT}). The first relation is obviously true if
$X_\psi=0$. Otherwise, it suffices to take into account the
relation $X_\tau=aX_\psi+bP$, $b=1/P^2$, which follows from
\eqref{eq:02-01}, and also the second relation in~\eqref{eq:02-01}.

(ii) (see, e.g.  \cite{Mas6}) Assume that the matrix
$X_\alpha-i\varepsilon P_ \alpha$ is degenerate. Then there exists
a vector $\xi\neq0$ such that $X_\alpha\xi=i\varepsilon P_
\alpha\xi$. The Lagrangian property implies that $P_ \alpha^T
X_\alpha-X_ \alpha^T P_\alpha=0$, where the symbol~$T$ stands for
transposition. It follows that $P_ \alpha^T X_\alpha\xi-X_ \alpha^T
P_\alpha\xi= i(\varepsilon P_ \alpha^TP_
\alpha+\frac{1}{\varepsilon}X_ \alpha^TX_ \alpha)\xi=0$. The
matrices $P_ \alpha^TP_ \alpha$, $X_ \alpha^TX_ \alpha$ are
nonnegative, and hence the last relation holds if and only if $P_
\alpha\xi=X_ \alpha\xi=0$. It follows that the rank of $4\times 2$
matrix $\begin{pmatrix}P_ \alpha \\X_ \alpha\end{pmatrix}$ is less
than 2, which contradicts to assumption that the dimension of
$\Lambda^2$ is~2.
\end{proof}

\subsection{Maslov index of regular points and closed paths}

Fix a regular point $\alpha_0\in \Lambda^2$, which we call the
\textit{central point}. Without loss of generality, we assume that
$\cJ(\a_0)>0$. Next, let $\alpha\in \Lambda$ be an arbitrary
nonsingular point. Fix some path $\gamma(\alpha_0,\alpha)\in
\Lambda^2$ joining $\a_0$ with $\alpha$ and define the
\textit{Maslov index} of~$\a$ by the formula
\begin{equation}\label{Indreg0}
m(\alpha)=\frac{1}{\pi}\lim_{\varepsilon\to+0} \operatorname{Arg}
\mathcal{J}^\varepsilon\big|_{\alpha_0}^\alpha,
\end{equation}
where the increment of the argument is taken along
$\gamma(\alpha_0,\alpha)$. In practice, it is better to use the
integral formula
\begin{equation}\label{Indreg1}
m(\alpha)=\frac{1}{\pi}\lim_{\varepsilon\to+0}
\im\int_{\gamma(\alpha_0,\alpha)}
\frac{d\,\mathcal{J}^\varepsilon}{\mathcal{J}^\varepsilon}.
\end{equation}
The index $m(\alpha)$ is an integer \textit{depending on the choice
of the path} $\gamma(\alpha,\alpha_0)$ (and remaining constant
under continuous deformations of the path). In particular,
$m(\a_0)=0$ provided that for the path joining~$\a_0$ with itself
one takes a path homotopic to the trivial path (which does not
leave~$\a_0$ at all).

Let $\gamma$ be some closed path on~$\La^2$; then we can define the
\textit{Maslov index $m(\gamma)$ of} $\gamma$ by setting
\begin{equation}\label{Indclosed}
\operatorname{ind} \gamma= \frac{1}{\pi}\lim_{\varepsilon\to+0}
\operatorname{Arg}_\gamma
\mathcal{J}^\varepsilon\big|_{\alpha_0}^{\alpha_0}=
\frac{1}{\pi}\lim_{\varepsilon\to+0} \operatorname{Im} \oint_\gamma
\frac{d\,\mathcal{J}^\varepsilon}{\mathcal{J}^\varepsilon} =
\frac{1}{\pi i} \oint_\gamma
\frac{d\,\mathcal{J}^\varepsilon}{\mathcal{J}^\varepsilon},\qquad\e>0.
\end{equation}

\begin{example}\label{3.1}
Let $\La^2$ be the Lagrangian manifold~\eqref{eq:01-1}. Fix the
central point $\alpha_0=(\tau=\delta,\psi=0)$, $\delta\to+0$ (i.e.,
$\delta=+0$). One can readily find that
\begin{equation*}
 \mathcal{J}^\varepsilon=\det\begin{pmatrix}
\cos\psi & (\tau-i\varepsilon)\sin\psi \\-\sin\psi&
(\tau-i\varepsilon)\cos\psi
\end{pmatrix}=(\tau-i\varepsilon).
\end{equation*}
Thus, the points with nonzero coordinate~$\tau$ are regular. For
the index~$m(\alpha)$, we have
\begin{equation*}
m(\alpha)=\frac{1}{\pi}\lim_{\varepsilon\to 0}\operatorname{Im}
\int_{(\delta,0)}^{(\tau,\psi)}\frac{d\tau} {\tau-i\varepsilon}
=\frac{1}{\pi}\lim_{\varepsilon\to
0}\Bl(\arctan\Bl(\frac{\tau}{\varepsilon}\Br)-
\arctan\Bl(\frac{\delta}{\varepsilon}\Br)\Br).
\end{equation*}
The last expression is $0$ if $\tau>0$ and $-1$ if $\tau<0$ for any
positive $\delta$. Thus, $m(\a(\tau,\psi))=0$ for $\tau>0$ and
$m(\a(\tau,\psi))=-1$ for $\tau<0$. By the way, note that the index
is independent of the choice of the path in this example.
\end{example}

\subsection{Nonsingular and singular charts.
Canonical atlas}\label{sec:02-04aa}

Maslov's canonical operator $K=K^h_{(\La^2,\mu)}$ associates a
rapidly oscillating function $u(x,h)=[K^h_{(\La^2,\mu)} A](x,h)$ to
every function $A\in C_0^\infty(\La^2)$.\footnote{The condition
that $A$ is compactly supported is convenient when describing the
general theory. If the projection $\pi_x\colon\La^2\lra
\mathbb{R}_x^2$ is proper (i.e., the preimage of every compact set
is compact), then one can safely replace $C_0^\infty(\La^2)$ with
$C^\infty(\La^2)$. This is what we do in our examples, where the
function on which the canonical operator acts is \textit{not}
compactly supported.} It is convenient to split the definition of
the canonical operator into two parts, local and global. In the
local definition, the manifold $\La^2$ is covered by special
connected simply connected domains, which we will refer to as
\textit{canonical charts}, and the canonical operator is defined
separately in each chart (i.e., on functions supported in that
chart). To pass to the global definition, the local canonical
operators are compared on the intersections of charts and pasted
together with the help of a partition of unity.

There are two kinds of canonical charts, nonsingular and singular.

\textbf{Nonsingular charts.} A point $\a\in\La^2$ is said to be
\textit{nonsingular} if $\cJ(\a)\ne0$. Accordingly, a
\textit{nonsingular chart} is an arbitrary connected simply
connected domain $U\subset\La^2$ consisting of nonsingular points.
Since $\cJ(\a)\ne0$, it follows that there exists a smooth solution
$\alpha(x)=(\tau(x),\psi(x))$ of the system of equations
\begin{equation}\label{Solx}
   X(\alpha)=x\Longleftrightarrow X(\tau,\psi)=x.
\end{equation}
By solving this system, one passes from the coordinates
$\alpha=(\tau,\psi)$ on $\La^2$ to the coordinates $x=(x_1,x_2)$ on
the configuration space.

\textbf{Singular charts.} By definition, nonsingular charts exhaust
all $\La^2$ except for the focal (or singular) points where
$\cJ(\a)=0$. Near the focal points, we need a different kind of
charts. Let $\a^*\in\La$ be a focal point. Take some system
$(\tau,\psi)$ of eikonal coordinates on $\La$ in a neighborhood of
$\a^*$. The coordinates of $\a^*$ will be denoted by
$(\tau^*,\psi^*)$. Consider the equation
\begin{equation}\label{eq:02-08}
    \langle P(\tau,\psi),x-X(\tau,\psi)\rangle=0.
\end{equation}
\begin{lemma}
Equation~\eqref{eq:02-08} defines a smooth function
\begin{equation}\label{e04a}
    \tau=\tau(x,\psi)
\end{equation}
in a neighborhood of the point $(x^*,\psi^*)\in\RR^3$, where
$x^*=X(\tau^*,\psi^*)$, such that $\tau^*=\tau(x^*,\psi^*)$.
\end{lemma}
\begin{lemma}
There exists a neighborhood $W$ of the point $(x^*,\psi^*)\in\RR^3$
such that the following conditions hold\rom:

\rom{(i)} The differential $d(\tau_{\psi})$ is nonzero at each
point of the set
\begin{equation}\label{eq:02-07}
    \Pi=\{(x,\psi)\in W\colon \tau_{\psi}(x,\psi)=0\},
\end{equation}
which is therefore a smooth two-dimensional surface.

\rom{(ii)} The mapping $(x,\psi)\longmapsto (x,\tau_x(x,\psi))$ is
a diffeomorphism of $\Pi$ onto a neighborhood $U\subset\La$ of the
point $\a_0$ in $\La$.

\rom{(iii)} One has
$\det(P,P_\psi)
    \big|_{\tau=\tau(x,\psi)}\ne0,\qquad (x,\psi)\in W$.
\end{lemma}
The domain $U\subset \La^2$, together with the eikonal coordinates
$(\tau,\psi)$ and the function \eqref{e04a} defined on $W$, will be
called a \textit{singular chart} on $\La^2$. Without loss of
generality, we assume that $U$ and $W$ are connected and simply
connected.

\textbf{Canonical atlas.} It follows from the preceding that
$\La^2$ can be covered by nonsingular and singular charts. Let us
take and fix a locally finite cover $\La^2=\bigcup_{j}U_j$ of
$\La^2$ by nonsingular and singular charts. We assume that the
intersection of two arbitrary sets $U_j$ is connected and simply
connected. (Of course, it can in particular be empty.) Such a cover
is called a \textit{canonical atlas}, and we only deal with these
charts $U_j$, which will be called \textit{canonical charts}.
Without loss of generality, we assume that there exist eikonal
coordinates in each of the charts $U_j$.\footnote{In many physical
problems, eikonal coordinates can be introduced globally on the
entire~$\La^2$; see examples above and below.}

\subsection{Canonical operator
in a nonsingular chart}\label{sec:02-04}

In this case, the definition of the canonical operator coincides
with the standard one. Let $U_j\subset\La^2$ be a nonsingular
chart. It readily follows from the implicit function theorem that
$U_j$ is diffeomorphically projected onto some open subset of
$\mathbb{R}_x^2$, and so the variables $x=(x_1,x_2)$ can be used as
local coordinates in $U_j$ and $\a=\a(x)$ can be defined as the
solution of Eqs.~\eqref{Solx}.

Choose an eikonal $\tau$ and an additional coordinate~$\psi$
in~$U_j$. We define the \textit{Maslov index~$m_j$ in $U_j$} by
setting $m_j=m(\alpha_j)$. Now we construct the canonical operator
$K_j$ on functions $A_j(\alpha)=A_j(\tau,\psi)\in C_0^\infty(U_j)$
in nonsingular charts by the formula
\begin{equation}\label{eq:02-05}
    [K_j A_j](x,k)=
    \frac{e^{\frac{i}{h}\tau(\a)-i \frac{\pi m_j}{2}} A_j(\a)}
                  {\sqrt{\abs{\mathcal{J}(\a)}}}\bigg|_{\a=\a(x)}\equiv e^{\frac{i}{h}\tau(x)-i \frac{\pi m_j}{2}}A_j(\tau,\psi)
    \sqrt{\frac{{\mu(\tau,\psi)}\abs{P(\tau,\psi)}}
    {\abs{X_\psi(\tau,\psi)}}}
    \bigg|_{\substack{\tau=\tau(x)\\\psi=\psi(x)}}.
\end{equation}

\subsection{Canonical operator
in a singular chart}\label{sec:02-05}

Now let $U_j\subset\La^2$ be a singular chart with eikonal
coordinates $(\tau,\psi)$. Consider the Jacobians $\cJ$ and
$\tilde{\cJ}=\det(P(\tau,\psi),P_\psi(\tau,\psi))$ at some
nonsingular point $(\tau^j,\psi^j)$. The second expression is
nonzero everywhere in $U_j$, in contrast to $\cJ(\tau,\psi)$  which
can change (and usually changes) its sign. \textit{We define the
Maslov index} $m_j$ \textit{of the singular chart} $U_j$ \textit{by
setting} $m_j=m(\alpha_j)$ \textit{if $\cJ(\tau,\psi)=\det\frac{\pa
X}{\pa(\tau,\psi)}$ and $\tilde {\cJ}(\tau,\psi)=\det(P,P_\psi)$
have the same sign and} $m_j=m(\alpha_j)+\pi$ \textit{if they have
opposite signs}.

\begin{example}\label{3.2}
Consider the manifold \eqref{eq:01-1} again. For the singular chart
$U_{sing}$ we take a neighborhood of the circle $\tau=0
\,\Leftrightarrow\, p=\mathbf{n}(\psi),x=0$ defined by the
inequality $|\tau|<\mu,$ where $\mu$ is a positive number. We have
$\mathcal{J}=\tau$ and $\tilde{\mathcal{J}}=\det(P,P_{\psi})=1$.
Thus, if we take a point $\a=(\tau,\psi)$ with positive $\tau$,
then the signs of~$\mathcal{J}$ and~$\tilde{\mathcal{J}}$ coincide
and $m(U_{sing})=0$. It is easily seen that the result will be the
same if one takes a point $\a=(\tau,\psi)$ with negative $\tau$.

Now we define the action of the canonical operator in the singular
chart $U_j$ on a function $A_j\in C_0^\infty(U_j)$ by the formula
\begin{equation}\label{eq:02-10}
    [K_j A_j](x,h)=\BL(\frac{i}{2\pi h}\BR)^{1/2}e^{-i\frac{\pi m_j}{2}}
    \int _{\mathbb{R}}\left[e^{\frac{i}{h}
    \tau }A_j(\tau,\psi)\sqrt{\mu \abs{\det(P,P_\psi)}}
      \right]_{\tau=\tau(x,\psi)}
      d\psi,
\end{equation}
where $\arg i=\pi/2$.
\end{example}

\begin{theorem}\label{th000}
The singular canonical operator~\eqref{eq:02-10} coincides modulo
$O(h)$ with the standard Maslov canonical operator
\rom{\cite{Mas1,MaFe1}}. In particular, the Maslov index of the
singular chart~$U_j$ coincides modulo~$4$ with the Maslov index of
the corresponding canonical chart in the standard construction of
the canonical operator.
\end{theorem}

The proof of this theorem will be given in
Appendix~\ref{sec:02-06}.

\subsection{Quantization
conditions and the global definition of the canonical
operator}\label{sec:02-07}

To define the canonical operator globally, we need to choose the
local eikonals $\tau_j$ and numbers $m_j$ (or, equivalently, the
arguments of the Jacobians $\cJ$ and$\wt\cJ$) in the canonical
charts in such a way that the local canonical operators coincide
modulo $O(h)$ on the intersections of their canonical charts. This
is possible if the \textit{Bohr--Sommerfeld quantization
conditions}
\begin{equation}\label{quant-cond}
 \frac2{\pi h}\oint_{\ga_j}P(\a)\,dX(\a)\equiv\ind\ga_j  \pmod4,
 \qquad j=1,\dotsc,N,
\end{equation}
are satisfied for a basis $\ga_1,\dotsc,\ga_N$ of independent
cycles on $\La^2$. Here $\ind\ga_j$ is the Maslov index of the
cycle $\ga_j$ and $N$ is the Betty number of~$\La^2$. For the
manifold in Sec.~2, $N=1$ and $\ind\gamma_1=0$.

Let conditions \eqref{quant-cond} be satisfied, and let
$\{\mathbf{e}_j\}$ be a locally finite partition of unity on
$\La^2$ subordinate to the cover $\{U_j\}$. Let us define the
global canonical operator $K^h_{\La^2}$ acting on smooth
functions~$A(\alpha)$ on $\La^2$ by the formula
\begin{equation}\label{global-operator}
    u(x,h)= K^h_{\La^2} A\equiv\sum_{j}K_j(\mathbf{e}_jA),
\end{equation}
where the sum is taken over all charts~$U_j$ of the canonical atlas
on~$\La^2$.

\begin{theorem}\label{th030}
If the quantization conditions are satisfied, then the canonical
operator~$K^h_{\La^2}$ defined in~\eqref{global-operator} is modulo
$O(h)$ independent of the choice of the charts~$U_j$ and the
partition of unity $\mathbf{e}_j$.
\end{theorem}

\begin{proof}
This theorem follows from Theorem~\ref{th000} and from the fact
that the desired assertion holds for the ``standard'' canonical
operator.
\end{proof}

A practical consequence of this theorem is as follows. Assume one
wishes to construct the asymptotic solution in a neighborhood of
some given point $x$. Then one should find all points $\alpha_k(x)$
on $\Lambda$ such that $X(\alpha_k(x)=x,\quad k=1,\ldots,k_0$. If
all points $\alpha_k(x)\in \Lambda^2 $ are ``far'' from the focal
points, then the sum consists of the WKB type
solutions~\eqref{eq:02-05}. If one or several (or an infinite set
of) focal points $\alpha^*_n=(\tau^*_n,\psi^*_n)$ on $\Lambda$ are
close to or even coincide with some $\alpha_{k'(x)}$, then the sum
expressing the function $u(x)$ should include the integrals
\eqref{eq:02-10} with cut-off functions $\mathbf{e}_n(\psi)$ such
that their support cover the $\varepsilon$-neighborhood of these
focal points. This, in turn, often gives the opportunity to express
these integrals via the special functions. We discuss such
simplifications later in Sec.~\ref{sec:02-09}. \textit{Yet another
useful corollary is that in specific computations one need not
assume that the domains \rom(charts\rom) $U_j$ are simply
connected.} Then, strictly speaking, the $U_j$ are no longer
charts, but the formulas remain valid; moreover, some functions in
the partition of unity are aggregated, and the corresponding
formulas are simplified greatly.

\subsection{Relationship with differential equations.
Commutation of the canonical operator with  $h$-differential and
$h$-pseudodifferential operators} \label{sec:02-08}

To make the exposition self-contained, let us present a well-known
assertion providing the application of the canonical operator in
partial differential equations. Consider a differential or
pseudodifferential operator with a small parameter~$h$,
\begin{equation}\label{e524}
    \wh L=L(\2x,\1{\wh p},h)\equiv L\BL(\2x,-ih\1{\pd{}x},h\BR),
\end{equation}
given by its symbol $L(x,p,h)$ with the Taylor expansion
$L(x,p,h)=H(x,p)+hL_1(x,p)+h^2L_2(x,p)+\dotsm$. Recall that
$H(x,p)$ is called the \textit{principal symbol} of the operator
$\wh L$, or the \textit{classical Hamiltonian}, and the function
$\frac12(\tr H_{xp})-iL_1(x,p)$, where $\tr H_{xp}$ is the trace of
the matrix $H_{xp}(x,p)$, is called the \textit{subprincipal
symbol} of the operator $\wh L$.
\begin{theorem}\label{th51}
Let $A\in C_0^\infty(\La^2)$. Then
\begin{equation}\label{e525}
    \wh L (K^h_{\La^2} A)=K_\La^2(H|_{\La^2}A)+O(h),
\end{equation}
where $H|_{\La}$ is the restriction of $H(x,p)$ to~$\La^2$. If,
moreover, $H(x,p)\equiv0$ on $\La^2$ and the measure $d\mu$ is
invariant with respect to shifts along the trajectories of the
Hamiltonian vector field
\begin{equation}\label{e526}
    \od{}t=H_p(x,p)\pd{}x-H_x(x,p)\pd{}p,
\end{equation}
then
\begin{equation}\label{e527}
    \wh L(K^h_{\La^2} A)=-ihK^h_{\La^2}\BL(\od A t-
    \frac12(\tr H_{xp})|_{\La^2}A+iL_1|_{\La^2}A\BR)+O(h^2).
\end{equation}
\end{theorem}
\begin{proof}
This is a consequence of Theorem~\ref{th000} and the theorem on the
commutation of the ``standard'' canonical operator with
differential operators.
\end{proof}

In particular, this theorem gives a recipe for constructing
Lagrangian manifolds, including manifolds with eikonal coordinates
(see the examples below).

\subsection{Simplification of solutions near the caustics}
\label{sec:02-09}

A straightforward application of formula~\eqref{eq:02-10}
specifying the canonical operator~$K^h_{\La^2}$ in the singular
chart~$U_j$ with eikonal coordinates $(\tau,\psi)$ requires
computing an integral of a rapidly oscillating function. It is
natural to ask whether one can avoid this labor-consuming
computation. We already know that if the support of the amplitude
$A_j(\tau,\psi)$ in~\eqref{eq:02-10} does not meet the set
$\Ga\subset\La^2$ of focal points, then formula~\eqref{eq:02-10}
can be reduced to the form~\eqref{eq:02-05}, which does not contain
integration. Hence, by using a partition of unity, one can readily
verify that it suffices to study the problem for the case in which
the support $\supp A_j$ is contained in a small neighborhood of
some focal point $\a^*=\a(\tau^*,\psi^*)\in\Ga$. Then the
function~\eqref{eq:02-10} is modulo $O(h^\infty)$ concentrated in a
neighborhood of the \textit{caustic} (which is the
projection~$\pi_x(\Ga)$ of the set $\Ga\in\La^2$ of focal points
of~$\La^2$ onto~$\RR_x^2$), or, more precisely, in a neighborhood
of the projection $x^*=X(\tau^*,\psi^*)$ of the point~$\a^*$. We
wish to simplify the integral formula~\eqref{eq:02-10}
for~$K^h_{\La^2}$ in a neighborhood of the caustic by expressing
the asymptotics of the function $K_jA_j$ via known special
functions.

The asymptotic expansion of the integral~\eqref{eq:02-10} is
related to the stationary (critical) points~$\psi_{cr}(x)$ of the
phase function $\tau(x,\psi)$ of this integral. These points depend
on the variables $x=(x_1,x_2)$ and prove to be degenerate if
$x$~lies on the caustic. It follows from the catastrophe theory and
the stationary phase method \cite{MaFe1, Arn2} that (except for
``superdegenerate'' cases) only a small neighborhood of the point
$\psi^*=\psi_{cr}(x^*)$ contributes to the asymptotic expansion of
the integral~\eqref{eq:02-10} for~$x$ close to~$x^*$, and the
contribution is related to the \textit{normal form} of the phase
function~$\tau$. This form is determined by the first nonzero
coefficient in the Taylor series expansion of $\tau-\tau^*$ in
powers of $\psi-\psi^*$. To obtain an asymptotics uniform in
$(x_1,x_2)$ near the caustic in some neighborhood independent
of~$h$, one needs to reduce the phase function to a normal form in
this neighborhood. This procedure is based on the Malgrange
preparation theorem and hence is not constructive.

More constructively, one can replace the phase function~$\tau$ by a
finite segment of its Taylor series and obtain the asymptotics of
the integral in an $O(h^\delta)$-neighborhood of the caustic for
some~$\delta$; for the points~$x$ near the boundary of this
neighborhood, both this asymptotics and the WKB
representation~\eqref{eq:02-05} hold. Moreover, one can replace the
nonlinear dependence of~$\psi_{cr}$ and~$\psi$ on the variables~$x$
by a linear approximation in $x-X(\tau^*,\psi^*)$ and set
$x=X(\tau^*,\psi^*)$ in the amplitude~$A$. This provides a complete
description of the asymptotic solution of the original problem.
\textit{Note that we do not paste together various asymptotics of
solutions as in the method of matched asymptotic expansions. We
only simplify Maslov's canonical operator in various domains.} In
general position, the sufficient Taylor polynomial is of degree~$3$
for edges of the caustics (this results in the Airy function)
and~$4$ for the cusp of the caustic (this results in the Pearcy
function). There is a vast literature on the Airy and Pearcy
functions and their applications to ray expansions and the
semiclassical approximation; here we only note
\cite{Berry,Pearcy,KravtsovOrlov}.

Despite being simple in principle, the construction itself, as well
as its justification, involves quite a few technical details, and
we give these in Appendix~2.

\section{EXAMPLES}\label{sec:03}

Let us discuss typical situations in which condition~\ref{co:02-1}
for the existence of eikonal coordinates is satisfied.

\begin{example}\label{ex1}
Let the Hamiltonian $H(x,p)$ be a homogeneous function of degree
$1$ with respect to the variables $p$ (i.e., $H(x,\la p)=\la
H(x,p)$ for $\la>0$), and suppose that a Lagrangian manifold
$\La^2$ lies in the level set $\{(x,p)\colon H(x,p)=1\}$. Then
condition \ref{co:02-1} is satisfied. Indeed, since
$H|_{\La^2}=\const$, it follows that $\La^2$ is invariant with
respect to the Hamiltonian vector field $V(H)=H_p\pa_x-H_x\pa_p$,
and $P\,dX(V(H)|_{\La^2})=PH_p(X,P)=H_{\La^2}=1$ by the Euler
identity, so that the form $P\,dX$ is necessarily nonzero. We also
see that the Hamiltonian vector field does not vanish on $\La^2$.
Hence we can introduce a local coordinate system $(\tau,\psi)$ on
$\La^2$ such that $\psi$ is constant along the trajectories and
$\tau$ is the time along the trajectories (the so-called
\textit{proper time}); these coordinates are eikonal coordinates.

This gives the following (well-known) method for constructing
Lagrangian manifolds with eikonal coordinates. Namely, let
$\La^1_0=\{(x,p)\in\mathbb{R}^4_{p,x}\colon x=X^0(\psi),
p=P^0(\psi)\}$ be a smooth open ($\psi\in \mathbb{R}$) or closed
($\psi\in \mathbb{S}$) curve in the four-dimensional phase space
such that $H|_{\La^1_0}=1$ and \textbf{(i)} each trajectory
$(P(\tau,\psi),X(\tau,\psi))$ of the Hamiltonian system
\begin{equation}\label{HamSys}
 \dot p=-H_x,\qquad \dot x=H_p
\end{equation}
issuing from $\La_0^1$ is transversal to it and \textbf{(ii)} (see
\cite{Vainberg,Kucherenko}) all these trajectories leave each
bounded domain in $\mathbb{R}^4_{x,p}$ in finite time. Then the
union of these trajectories is the Lagrangian manifold
$\La^2=\{(x,p)\in\mathbb{R}^4_{x,p}\colon x=X(\tau,\psi),
X(\tau,\psi)\}$ with eikonal coordinates $(\tau,\psi)$. The fact
that $\La$ is a smooth manifold follows from assumptions
\textbf{(i)} and \textbf{(ii)} (see \cite{Vainberg,Kucherenko}).
The Lagrangian property follows from the conservation of the skew
inner product of solutions of linear Hamiltonian systems.
\end{example}

\begin{example}[Maupertuis--Jacobi
principle and canonical operator]\label{ex2}

Consider the Ha\-mil\-tonian
\begin{equation*}
  \cH(x,p)=\frac{p^2}{2\mathbf{m}}+v(x),
\end{equation*}
arising in particular when constructing semiclassical asymptotics
for the Schr\"odinger equation. Here $\mathbf{m}$ is the mass, and
the potential $v(x)$ is assumed to be smooth and bounded. Suppose
that a Lagrangian manifold $\La$ lies in the level set
$\{(x,p)\colon \cH(x,p)=E\}$ for some $E>\max_x v(x)$ (unbound
states). This property can be rewritten in the form
\begin{equation*}
  \frac{p^2}{2\mathbf{m}}=E-v(x)\quad\Longleftrightarrow\quad \frac{p^2}{2\mathbf{m}(E-v(x))}=1
  \quad\Longleftrightarrow\quad\frac{|p|}{\sqrt{2\mathbf{m}(E-v(x))}}=1.
\end{equation*}
By the Maupertuis--Jacobi principle (see~\cite{Arn1}, and
also~\cite{DobRou10}), the trajectories  of the field~$V(\cH)$
on~$\La$ coincide (up to a change of time) with those of the
field~$V(H)$ corresponding to the Hamiltonian homogeneous of first
order with respect to~$p$ (defining the so-called Finsler metric):
\begin{equation*}
  H(x,p)={\abs{p}}C(x),\quad C(x) =1/{\sqrt{2\mathbf{m}(E-v(x))}}.
\end{equation*}
Namely, let  $(P(\tau,\psi),X(\tau,\psi))$ be the solutions of the
system
\begin{equation}\label{HamSys1}
\dot p=-|p|C_x,\qquad \dot x=\frac{p}{|p|}C
\end{equation}
issuing from the corresponding curve $\Lambda^1_0$. Then, according
to Maupertuis--Jacobi principle, we can introduce a new time
$t=t(\tau,\psi)$ by setting
\begin{equation}\label{time}
    t=\int_0^\tau\sqrt{2\mathbf{m}(E-V(X(\tau,\psi)))}\,d\tau=
    \int_0^\tau|P(\tau,\psi)|\,d\tau.
\end{equation}
By solving this equation for $\tau$, we obtain a function
$\tau(t,\psi)$. Now, by substituting it into the functions
$P(\tau,\psi)$ and $X(\tau,\psi)$, we obtain the functions
$(\mathcal{P}(t,\psi),\mathcal{X}(t,\psi)=
(P(\tau(t,\psi),\psi),X(\tau(t,\psi),\psi))$, which are the
solutions of the Hamiltonian system with Hamiltonian $\mathcal{H}$.
On the manifold $\La^2$, one has $H(x,p)=1$, and we are in the
framework of Example~\ref{ex1}. Thus, on $\La^2$ we have two
coordinate systems (two distinct parametrizations) $(t,\psi)$ and
$(\tau,\psi)$ related by the formulas $\tau=\tau(t,\psi)$,
$\psi=\psi$, and the Jacobian of passage from $(t,\psi)$ to
$(\tau,\psi)$ is equal to $1/|P(\tau,\psi)|$.

Let us explain the passage from the Hamiltonian system with the
$\cH(x,p)$ to the system with the Hamiltonian $H(x,p)$ and from the
coordinates $(t,\psi)$ to $(\tau,\psi)$ at a ``quantum level'' in
the following way. Consider the stationary Schr\"odinger equation
\begin{equation}\label{Schr}
  -\frac{h^2}{2\mathbf{m}}\triangle u+vx)u=E u,
\end{equation}
under the assumption that the potential $v(x)$ is a smooth function
such that $E>V$ and $h$ is a small positive parameter.  It is
useful to recall that, by introducing the function
$n^2=2\mathbf{m}(E-v(x))$ and large parameter $k=1/h$, we can
rewrite \eqref{Schr} in a form of the Helmholtz equation
\begin{equation}\label{eq:03-1}
(\Delta+k^2n^2(x))u(x,k)=0, \qquad x\in\mathbb{R}^2,
\end{equation}
where the refraction coefficient $n(x)$ is a smooth everywhere
positive function.
\begin{remark}
It is often assumed in physical applications that $v(x)\to 0$ as
$|x|\to \infty$. In this case, we can assume that $n^2(x)\to 1$ as
$|x|\to \infty$. To this end, in the preceding formulas one should
replace the parameter~$h$ by the parameter $h'=h/\sqrt{E}$ and the
potential~$v$ by the potential~$v'=v/E$.
\end{remark}

Set $C(x)=1/\sqrt{2\mathbf{m}(E-v)}\equiv 1/n(x)$. By dividing
\eqref{eq:03-1} by $2\mathbf{m}(E-v)$, we obtain
\begin{equation}\label{eq:03-1a}
  (-{h^2}{C^2(x)}\Delta-1)u(x,h)=0, \qquad x\in\RR^2,
\end{equation}
The operator $-{h^2}{C^2(x)}\Delta
=\stackrel{2}{C^2(x)}\stackrel{1}{\hat p^2}$, $\hat
p=-ih\frac{\pa}{\pa x}$, is essentially self-adjoint in the
weighted $L_2$ space with weight $1/C^2(x)$ and is positive in this
space. Hence there exists a self-adjoint positive operator $\hat
L=\sqrt{-{h^2}{C^2(x)}\Delta}=\sqrt{\stackrel{2}{C^2(x)}\stackrel{1}{\hat
p^2}}$ and one can rewrite \eqref{eq:03-1a} in the form $(\hat
L+1)(\hat L-1)u=0$. The first factor on the left-hand side is a
positive operator, and hence one can divide this equation by it and
rewrite it in the form $(\hat L-1)u=0$. Now let us represent the
operator $\hat L=\sqrt{\stackrel{2}{C^2(x)}\stackrel{1}{\hat p^2}}$
in the form of a pseudodifferential operator
$L(\stackrel{2}x,\stackrel{1}{\hat p},h)$ with symbol $L(x,p,h)$
having the asymptotic expansion $L(x,p,h)=H(x,p)+h L_1(x,p)+h^2
L_2(x,p)+\dotsm$. One cannot find the exact symbol $L(x,p,h)$, but
it is easy to find the coefficients $H$ and $L_j(x,p)$. For
constructing the leading tern of the semiclassical asymptotics,
only $H$ and $L_1$ are important. Let us find them. We have
$(L(\stackrel{2}x,\stackrel{1}{\hat
p},h))^2=\stackrel{2}{C^2(x)}\stackrel{1}{\hat p^2}$. Using the
formulas in \cite{Mas6,MaNa2} and passing in this equation from
operators to their symbols, we find the equation for $L$:
\begin{equation*}
  L(\stackrel{2}{x},p-ih \stackrel{1}{\frac{\pa}{\pa x}},h)L(x,p,h)
 =C^2(x)p^2.
\end{equation*}
Using the Taylor expansion with respect to $ih
\stackrel{1}{\frac{\pa}{\pa x}}$ of the first factor, we readily
obtain $H(x,p)\equiv L_0(x,p)=C(x)|p|$ and
$L_1(x,p)=\frac{i}{2H}\frac{\pa H}{\pa p}\frac{\pa H}{\pa x}\equiv
i\frac{\langle C_x,p\rangle}{2 |p|}\equiv
\frac{i}{2}\rm{tr}H_{px}$.

Assume that we have constructed a Lagrangian manifold as in
Example~\ref{ex1} invariant with respect the corresponding vector
field $V(H)$. Here the Hamiltonian is $H=|p|C(x)$, and the
Hamiltonian system has the form~\eqref{HamSys1}.

Let $A(\tau,\psi)$ be an amplitude on $\La$. Consider the function
$u=K^h_{\La^2} A$.  Let us find out under what conditions this
function is a solution of Eq.~\eqref{Schr} modulo $O(h^2)$.
According to the construction of~$\La^2$, the derivative
$\frac{d}{d t}$ coincides with the derivative $\frac{d}{d \tau}$ on
$\La$, and hence, by applying Theorem~\ref{th51}, we obtain the
equation $\frac{d A}{d\tau}-\rm{tr}H_{px}|_{\La^2}A=0$ for the
amplitude. By the Hamiltonian system, we have $\dot P=-|P|C_x(X)$,
which gives $\rm{tr}H_{px}|_{\La}=\frac{\langle C_x,p\rangle}{
|p|}|_{\La^2}=-\frac{\langle\dot P,P\rangle}{|P|}=\frac{|\dot
P|}{|P|}=-\frac{d \log |P|}{d \tau}$. Thus, the solution $A$ can be
presented in the form
\begin{equation}\label{Sol1}
    A=\frac{A_0(\psi)}{|P(\tau,\psi)|}
    =\frac{A_0}{\sqrt{2\mathbf{m}(E-v(X(\tau,\psi)))}}
\end{equation}
(because $\frac{|P|}{\sqrt{2\mathbf{m}(E-v(X(\tau,\psi)))}}\equiv
H=1$), where $A_0$~is an arbitrary smooth compactly
supported\footnote{The last property is not necessary; it only
guarantee that the function $u$ is well defined in an appropriate
function space.} function.

Now we can state the general result.

\begin{lemma}
Let $\La^2$ be an invariant Lagrangian manifold with eikonal
coordinates constructed from some smooth curve $\La^1_0$, and let
$A_0(\psi)$ be an arbitrary smooth function. Then the function
\begin{equation}\label{Schr1}
 u=K_{\La^2}\BL[\frac{A_0(\psi)}{|P(\tau,\psi)|}\BR]=
 \frac{1}{\sqrt{2\mathbf{m}(E-v(X(\tau,\psi)))}}K^h_{\La^2}[A_0(\psi)]+O(h)
\end{equation}
satisfies the Schr\"odinger equation \eqref{Schr} modulo $O(h^2)$.
\end{lemma}

The last formula and formula~\eqref{time} give a different
representation of Maslov's canonical  operator based on the
coordinates $(t,\psi)$ and hence a representation of the asymptotic
solution of Eq.~\eqref{Schr} different from the one used, in
particular, in~\cite{Kucherenko,Vainberg}.

The suggested method for constructing Lagrangian manifolds permits
one to obtain a broad class of asymptotic solutions of the
Schr\"odinger equation. The problem is to find appropriate curves
$\La^1_0$ and functions $A(\psi_0)$ giving asymptotic solutions
that are of interest from the viewpoint of applications. Note also
that problems with an axisymmetric potential can be reduced to
one-dimensional problems if one passes to polar coordinates.
However, this passage results in a singularity at the origin and
requires additional studies when constructing the corresponding
asymptotics. In our scheme, there are no ``coordinate''
singularities, and the two-dimensional nature of the problem can
readily be taken into account when specifying the corresponding
Lagrangian manifold, which, in our opinion, makes the suggested
scheme very efficient in applications; see the examples below.
\end{example}

\begin{example}[two Lagrangian manifolds
important in physical applications]\label{ex3a}

Let us \newline present two (well-known) special Lagrangian
manifolds with eikonal coordinates important in physical
applications. Let the potential $v(x)$ or velocity $C$ satisfy
assumptions \textbf{(i)} and \textbf{(ii)}. The first Lagrangian
manifold is related to the circle
\begin{equation}\label{circle}
    \La^1_0=\{(x,p)\in\mathbb{R}^4_{p,x}\colon x=0, p=n(0)\bn(\psi),
    \;\psi\in \mathbb{S}^1\},
    \qquad\bn(\psi)=\begin{pmatrix}\cos \psi\\
    \sin \psi\end{pmatrix}.
\end{equation}
Then $\La^2 =\{(p,x)\in \mathbb{R}^4_{p,x}\colon
p=(P(\tau,\psi),X(\tau,\psi))$, $\tau\in \mathbb{R}$, $\psi\in
\mathbb{S}^1\}$, is a smooth manifold diffeomorphic to the
two-dimensional cylinder in~$\mathbb{R}^4_{px}$. Maslov's canonical
operator on such manifolds gives the generalized asymptotic
eigenfunctions of the Schr\"odinger operator \eqref{Schr}, and the
canonical operator on the ``half-cylinder'' (with $\tau\geq 0)$
defines the asymptotics of the Green function of  the Schr\"odinger
operator \eqref{Schr} or the corresponding Helmholtz operator (see
\cite{Kucherenko}).

Now assume for simplicity that the potential $v$ is a compactly
supported function and its support $\supp V$ lies in some domain
$D\in \mathbb{R}^2_{x}$ in the half-plane $x_1<a$. Now take
$\La^1_0$ to be the straight line $\La^1_0=\{p_1=\sqrt{E},p_2=0,
x_1=a, x_2=\psi\} $ and set $A_0=1$. Then the function~$u$
\eqref{Schr1} gives the leading term of the asymptotic solution~$w$
of the scattering problem for Eq.~\eqref{Schr}. A precise statement
and proof can be found in \cite[Chap.~XI]{Vainberg}.
Formula~\eqref{Schr1} gives a different and, together
with~\eqref{eq:02-05}, \eqref{COMAiry}, and~\eqref{COMPearcy}, more
explicit representation of the corresponding asymptotics.
\end{example}

Now consider  more specific examples.

\begin{example}\label{ex3}
(Asymptotics of a solution of the Helmholtz equation with
axisymmetric refraction coefficient and of generalized
eigenfunctions of the Schr\"odinger equations with axisymmetric
potential). Consider the Helmholtz equation \eqref{eq:03-1} in
which the refraction coefficient $n(x)$ is a smooth everywhere
positive function depending only on $\abs{x}$, $n(x)\equiv
n(\abs{x})$, and equal to $1$ for $\abs{x}>R_0$. For this equation,
consider the problem of finding rapidly oscillating solutions whose
associated Lagrangian manifold $\La^2$ coincides with the manifold
\eqref{eq:01-1} for large $\abs{x}$. (In particular, if $A=1$ and
the measure $d\mu$ on $\La^2$ is chosen to be invariant with
respect to the Hamiltonian vector field, as is actually the case in
the example with $n^2(x)\equiv1$, considered in Sec.~\ref{sec:21},
then our solution will be in a sense nearly proportional to the
Bessel function $\mathbf{J}_0\bl(\frac{\abs{x}}{h}\br)$ for large
$\abs{x}$.)

For $\La^2$ we take the Lagrangian manifold passing through the
circle \eqref{circle} and invariant with respect to the Hamiltonian
vector field corresponding to the Hamiltonian
$\cH(x,p)=p^2-n^2(x)$. As was said above, to find $\La^2$ and
determine eikonal coordinates on $\La^2$, we use the
Jacobi--Maupertuis principle and proceed to the equivalent
first-order homogeneous Hamiltonian
$H(x,p)=\frac{\abs{p}}{n(\abs{x})}=\abs{p}C(\abs{x})$. The solution
of the corresponding Hamiltonian system can be sought in the form
$X(\tau,\psi)=\rho(\tau)\bn(\psi)$,
$P(\tau,\psi)=\cP_\rho(\tau)\bn(\psi)$. By solving the equations,
we find that
\begin{equation*}
    \rho(\tau)\quad\text{is the inverse function of }
    T(r)=\int_0^rn(r)\,dr,\qquad
    \cP_\rho(\tau)=n(\rho(\tau)).
\end{equation*}
The pair $(\tau,\psi)$ is an eikonal coordinate system on $\La^2$.
We see that for large $\abs{\tau}$ this manifold coincides with the
one considered in Sec.~\ref{sec:21}, but the parametrization is
different. (The two parametrizations differ for large $\tau$ by the
shift
\begin{equation}\label{phase-shift}
    \delta\tau=\int_0^{R_0}(n(r)-1)\,dr
\end{equation}
in the variable $\tau$.)

We claim that, in these eikonal coordinates, the entire manifold
$\La^2$ is covered by one singular canonical chart. Indeed,
Eq.~\eqref{eq:02-08} has the global solution
$\tau(x,\psi)=T(\langle x,\bn(\psi)\rangle)$. The set
$\Pi=\{(x,\psi)\colon\tau_\psi(x,\psi)=0\}$ has the form
$\Pi=\{(x,\psi)\colon x\perp \bn'(\psi)\}= \{(x,\psi)\colon
x\parallel \bn(\psi)\}$, and $d(\tau_\psi)\ne0$ on $\Pi$, because
$\pd{}x(\tau\psi)=n(\abs{x})\bn'(\psi)\ne0$ on $\Pi$. The mapping
$(x,\psi)\lra (x,\tau_x(x,\psi))$ acts on $\Pi$ by the formula
$\Pi\ni (x,\psi)\longmapsto (x,n(\abs{x})\bn(\psi))$
and is easily seen to be a diffeomorphism of $\Pi$ onto $\La^2$.
Finally,
\begin{equation*}
    \det(P,P_\psi)=\cP_\rho^2(\tau)\det(\bn,\bn_\psi)
    =\cP_\rho^2(\tau)=n^2(\rho(\tau))\ne0.
\end{equation*}
Thus, we see that the manifold $\La^2$ is indeed covered with one
singular canonical chart.

The measure invariant with respect to the Hamiltonian vector field
associated with the Hamiltonian $\cH(x,p)$ has the form
$d\mu=\frac1{n^2(\rho(\tau))}\,d\tau\wedge d\psi$. Thus,
\begin{equation}\label{otvet}
    [K_\La^2 A](x,h)=\BL(\frac{i}{2\pi h}\BR)^{1/2}
    \int_0^{2\pi}e^{\frac{i}{h}T(\langle x,\bn(\psi)\rangle)}
    A(\langle x,\bn(\psi)\rangle,\psi)\,d\psi.
\end{equation}
Let us show that this integral can be expressed via the Bessel
function. A straightforward computation proves the following
\begin{lemma}
There exists a smooth change of variables $y=g(x)$,
$\wt\psi=\wt\psi(x,\psi)$ such that $T(\langle
x,\bn(\psi)\rangle=\langle y,\bn(\wt\psi)\rangle$. Moreover,
\begin{equation}\label{zamena-1}
    g(x)=\frac{T(\abs{x})}{\abs{x}},\qquad
    \pd{\wt\psi}{\psi}\bigg|_\Pi=
    \sqrt{\frac{\abs{x}n(\abs{x})}{T(\abs{x})}}.
\end{equation}
\end{lemma}
Using this result and taking $A\equiv 1$, we obtain
\begin{equation}\label{result}
\begin{split}
    [K^h_{\La^2}1](x)&=\BL(\frac{i}{2\pi h}\BR)^{1/2}
    \int_0^{2\pi}e^{\frac{i}{h}\bl\langle\frac{T(\abs{x})}{\abs{x}} x,
    \bn(\wt\psi)\br\rangle}
    \BL(\pd{\wt\psi}{\psi}\BR)^{-1}\,d\wt\psi\\
    &=\BL(\frac{i}{2\pi h}\BR)^{1/2}
    \BL(\int_0^{2\pi}e^{\frac{i}{h}\bl\langle\frac{T(\abs{x})}{\abs{x}} x,
    \bn(\wt\psi)\br\rangle}
    \sqrt{\frac{T(\abs{x})}{\abs{x}n(\abs{x})}}\,d\wt\psi+O(h)\BR).
\end{split}
\end{equation}
The amplitudes in these two integrals coincide on $\Pi$;
consequently, their difference can be represented as the product of
a smooth function by $\tau_\psi$, and integration by parts shows
that the difference of the integrals is indeed $O(h)$ compared with
the main term. Now we finish the computation and find that
\begin{equation}\label{final}
    \BL(\frac{2\pi h}{i}\BR)^{1/2}[K_{\La^2}^h1](x,k)=
  a(\abs{x})
    \mathbf{J}_0\BL(\frac{T(\abs{x})}{h}\BR)+O(h),\quad a(\abs{x})=2\pi \sqrt{\frac{T(\abs{x})}{\abs{x}n(\abs{x})}}.
\end{equation}
We see that the canonical operator on this manifold gives a
``distorted'' Bessel function: it is multiplied by the factor
$a(\abs{x})$, which tends to unity as $\abs{x}\to\infty$, and,
which is more important, has the phase shift $\delta\tau$ given by
\eqref{phase-shift}: for $R>R_0$, Eq.~\eqref{final} has the form
\begin{equation*}
    \BL(\frac{2\pi h}{i}\BR)^{1/2}[K^h_{\La^2}1](x,h)=
    a(\abs{x})
   \mathbf{ J}_0(k(\abs{x}+\delta\tau))+O(h).
\end{equation*}
\end{example}

\begin{example}
Now consider a simple example from the wave beam theory. Consider
the following Cauchy problem for the Schr\"odinger type equation
arising in optics in the well-known paraxial approximation:
\begin{equation}\label{Schr1a}
  i h\frac{\pa v}{\pa t}
     =-i h c\frac{\pa v}{\pa x_3}
     - h^2\frac{ c}{4k}\BL(\frac{\pa^2 v}{\pa x_1^2}
     +\frac{\pa^2 v}{\pa x_2^2}\BR),\qquad
       v|_{t=0}=v_0.
\end{equation}
Here $c$ and $k$ are physical (positive) constants.  Let us make
the well-known change of variables $z=x_3-ct$ and set
$v(x_1,x_2,x_3,t)=u(x_1,x_2,z)$, where $u$ is the new unknown
function. The equation \eqref{Schr1} acquires the form of the 2D
Schr\"odinger equation
\begin{equation}\label{Schr2}
i h\frac{\pa v}{\pa t}= -\frac{h^2 }{2\mathbf{m}}\BL(\frac{\pa^2
v}{\pa x_1^2}+\frac{\pa^2 v}{\pa x_2^2}\BR),\quad
\mathbf{m}=\frac{2k}{c},\qquad u|_{t=0}=v_0(x_1,x_2,z),
\end{equation}
including the variable $z$ az the a parameter. Let $x$ be the
column 2-vector with components $(x_1,x_2)$. Needless to say, one
can solve the Cauchy problem for Eq.~\eqref{Schr2} by the Fourier
method and obtain the answer in the form of an integral of rapidly
varying functions, but the asymptotic simplification of such
integrals (by methods like the stationary phase method) is a rather
difficult problem, so that the approach based on Maslov's canonical
operator is in our opinion more efficient. Let us choose the
initial data in a form of the Maslov canonical operator on the
family of Lagrangian manifolds
$\La_0(z)=\{p=P^0(\alpha,\psi)\equiv\lambda(z)\mathbf{n}(\psi),
x=X^0\equiv \alpha\mathbf{n}(\psi),\alpha\in \mathbb{R},\tau\in
\mathbb{S}$ parametrized by~$z$,
\begin{equation}\label{Schr3}
    u|_{t=0}=v_0(x_1,x_2,z)=K_{\La_0^2(z)}[A].
\end{equation}
Here $\lambda(z)$ is a smooth function, for example,
$\lambda(z)=1/\sqrt{1+z^2}$ or $\lambda(z)=a+b(1+\tanh z)$, $a$ and
$b$ are constants, $a>0$, and $|b|<a$.

We take the measure density $\mu\equiv1$ and the amplitude
$A(\alpha)$ on $\La_0(z)$ of the form $A(\alpha)=g(\alpha)f(z)$,
where $g(\alpha)$ and $f(z)$ are smooth compactly supported
functions. For simplicity, assume that $g(\alpha)$ is even. Let us
fix the central point $\alpha_0=(\alpha=+0,\psi=0)$ on $\La_0(z)$.
Let us show that the solution \eqref{Schr2} connected with such
type of manifolds the generalizes the solutions known as Bessel
beams in optics.

Let us apply the general Maslov asymptotic construction of
solutions of the Cauchy problem for differential (and
pseudodifferential) equations. To  this end, we shift the manifold
$\Lambda_0(z)$ with the help of the phase flow $g^t_H$
corresponding to the Hamiltonian system with Hamiltonian
$H=\frac{p^2}{2\mathbf{m}}$. One can readily show that the
``shifted'' manifold is
$\La^2_t(z)=g^t_H\La_0(z)\equiv\{p=P(\alpha,\psi)
\equiv\lambda(z)\mathbf{n}(\psi), x=X^0+tH_p(P^0)\equiv
\bl(\alpha+t\frac{\lambda(z)}{\mathbf{m}}\br)\mathbf{n}, \alpha\in
\mathbb{R},\tau\in \mathbb{S}\}$. The central point on $\La^2_t(z)$
is the point $\alpha_0^t$ with coordinates $(\alpha=+0,\psi=0)$.
Then the leading term of the asymptotic solution of
problem~\eqref{Schr2}, \eqref{Schr3} is \cite{Mas1,MaFe1}
\begin{equation}\label{Schr4}
    u=v_0(x_1,x_2,z,t)=e^{i\frac{s(t)}{h}
    -i\pi m(\alpha_0^t)/2}K_{\La^2_t(z)}[A^t],
\end{equation}
where  $s(t)$ is the integral (the action) of the Lagrangian
$L=(\langle p,H_p\rangle-H)=\frac{p^2}{2\mathbf{m}}$ along the
trajectory formed by the central points,
\begin{equation*}
  s(t)=\int_0^t\frac{P^2(0,\psi)}{2\mathbf{m}}dt=
t\frac{\lambda^2(z)}{2\mathbf{m}}
\end{equation*}
and $m(\alpha_0^t)$ is the number of focal points on this
trajectory on the time interval $[0,t]$. One can readily show that
$m(\alpha_0^t)=0$. The amplitude $A^t$ is the solution of the
transport equation on $\La_t^2$, which has the form $\frac{\pa
A^t}{\pa t}=0$, and $A^t|_{t=0}=A_0$; hence $A^t=A(\alpha)$.

Now let us compute $K^h_{\La^2_t(z)}[A^t]$. The eikonal is
\begin{equation}\label{Eik}
    \tau=\int_{(\alpha=+0,\psi=0)}^{(\alpha,\psi)}P\,dX
    =\lambda(z)\alpha,
\end{equation}
and the eikonal coordinates are $(\tau=\lambda\alpha,\psi),
\alpha=\tau/\lambda$. In the eikonal coordinates, we have
$P=\lambda(z)\mathbf{n}(\psi),X=(\frac{\tau}{\lambda(z)}
+t\frac{\lambda(z)} {\mathbf{m}})\mathbf{n}(\psi)$ and
$X_\psi=(\frac{\tau}{\lambda(z)}+t\frac{\lambda(z)}
{\mathbf{m}})\mathbf{n}_\perp(\psi)=(\alpha+t\frac{\lambda(z)}
{\mathbf{m}})\mathbf{n}(\psi)$. We see that the manifolds
(cylinders in $4-D$ phase space $\mathbb{R}^4_{px}$)
$\Lambda^2_t(z)=\Lambda^2_0(z)$ coincide as geometrical objects,
and so $\La^2_0$ is invariant with respect to the phase flow
$g_t^H$. The Lagrangian singularities are again the circle
$\frac{\tau}{\lambda(z)}=-t\frac{\lambda^2(z)} {\mathbf{m}}$, which
moves with time on $\La_0$. Their projection onto the physical
plane is the point $x=0$ (which does not move with time $t$). One
can cover the manifold  $\La^2_0$ by two regular charts
$U_1=\{\tau>-t\lambda^2(z)+\delta\}$ and
$U_3=\{\tau<-t\lambda^2(z)-\delta\}$ and one singular chart
$U_2=\{|\tau+t\lambda^2(z)<2\delta\}$. The solution of the
equations $X(\tau,\psi)=x$ is $\tau=-t\frac{\lambda^2(z)}
{\mathbf{m}}\pm\lambda(z)|x|$, where the sign $+$ is taken for
$U_1$ (where $\tau>t\frac{\lambda^2(z)} {\mathbf{m}}$) and the sign
$-$ is taken for $U_2$ (where
$\tau<t\frac{\lambda^2(z)}{\mathbf{m}}$); $\mathbf{n}(\psi)=x/|x|$.
We specify the Maslov index $m_1=0$ for the points in $U_1$. By
analogy with Example~\ref{3.1}, we find that the Maslov index
$m_2=-1$ in $U_2$. Thus, outside some neighborhood of the origin
the canonical operator gives the formula
\begin{equation*}
 u=e^{-\frac{it\lambda^2(z)}{2\mathbf{m}h}}e^{-i\frac{\pi}{4}}
\frac{\sqrt{\lambda(z)}}{\sqrt{|x|}}
\Bigl(e^{\frac{i\lambda(z)|x|}{h}+\frac{i\pi}{4}}
g\Bl(|x|-t\frac{\lambda(z)}{\mathbf{m}}\Br)+
e^{-\frac{i\lambda(z)|x|}{h}-\frac{i\pi}{4}}
g\Bl(-|x|-t\frac{\lambda(z)}{\mathbf{m}}\Br)\Bigr) f(z).
\end{equation*}
Now let us represent the solution in a neighborhood of the origin
$x=0$. The solution $\tau(x,\psi,t)$ of the equation $\langle
P,x-X\rangle=0$ is $\tau=\langle x,\mathbf{n}(\psi)\rangle
-t\frac{\lambda^2(z)} {2\mathbf{m}}$, and
$\det(P,P_\psi)=\lambda^2(z)$. Since the last determinant does not
vanish anywhere on $\La_t(z)$, we can omit the cutoff function in
the singular chart and everywhere write
\begin{equation*}
  u=\BL(\frac{i}{2\pi h}\BR)^{1/2}e^{-\frac{it\lambda^2(z)}{\mathbf{m}h}}
\lambda(z) \int_0^{2\pi}e^{\frac{i}{h}\lambda(z)\langle
x,\mathbf{n}(\psi)\rangle} g\Bl({\langle
x,\mathbf{n}(\psi)\rangle}-t\frac{\lambda(z)}{2\mathbf{m}}\Br)f(z)\,d\psi.
\end{equation*}
Now, using the same argument as in the derivation
of~\eqref{eq:01-2c}, we obtain
\begin{multline*}
u=\pi \BL(\frac{i}{2\pi
h}\BR)^{1/2}e^{-\frac{it\lambda^2(z)}{2\mathbf{m}h}}
   \lambda(z)f(z)\BL[\BL(g\BL({|x|}-t\frac{\lambda(z)}{\mathbf{m}}\BR)+g\BL(-|x|-
   t\frac{\lambda(z)}{\mathbf{m}}\BR)\BR)\cJ_0\BL(\frac{\lambda(z)\abs{x}}{h}\BR)
    \\{}+
i\BL(g\BL(|x|-t\frac{\lambda(z)}{\mathbf{m}}\BR)-g\BL(-|x|-
   t\frac{\lambda(z)}{\mathbf{m}}\BR)\BR)
\cJ_1\BL(\frac{\lambda(z)\abs{x}}{h}\BR)\BR]\bigg|_{z=x_3-ct}.
\end{multline*}
\end{example}

\renewcommand{\sectionname}{APPENDIX}
\setcounter{section}{0}

\section{PROOF OF THEOREM~\ref{th000}}\label{sec:02-06}

We denote the product $A_j(\tau,\psi)\sqrt{\mu(\tau,\psi)}$ by
$\ph(\tau,\psi)$. Without loss of generality, we assume that the
support $\supp\ph$ is contained in a small neighborhood of a
singular point $(\tau_0,\psi_0)\in U_j$. Let us represent the
canonical operator on functions supported in a neighborhood of
$(\tau_0,\psi_0)$ in the form of the standard canonical operator in
a singular chart. Prior to this, we make a rotation of the
coordinate system on the $x$-plane (and the associated rotation of
the dual coordinate system on the $p$-plane)so as to ensure that
$P_2(\tau_0,\psi_0)=0$.
\begin{lemma}
Under the condition $P_2(\tau_0,\psi_0)=0$, the variables
$(x_1,p_2)$ can be taken for canonical coordinates on the
Lagrangian manifold in a neighborhood of the point
$(\tau_0,\psi_0)$.
\end{lemma}
\begin{proof}
We make all computations at $(\tau_0,\psi_0)$. Since $P_2=0$, it
follows that $\langle P,X_\tau\rangle\equiv P_1X_{1\tau}$ and hence
$X_{1\tau}\ne0$. Next, $0\ne\det(P,P_\psi)\equiv P_1P_{2\psi}$,
whence it follows that $P_{2\psi}\ne0$, and finally, $0=\langle
P,X_\psi\rangle\equiv P_1X_{1\psi}$, whence it follows that
$X_{1\psi}=0$, because $P_1\ne0$. Consequently,
\begin{equation*}
    \mathcal{J}_s\equiv\det\frac{(X_1,P_2)}{(\tau,\psi)}=
    \det\begin{pmatrix}
      X_{1\tau} & X_{1\psi} \\
      P_{2\tau} & P_{2\psi} \\
    \end{pmatrix}=\det\begin{pmatrix}
      X_{1\tau} & 0 \\
      P_{2\tau} & P_{2\psi} \\
    \end{pmatrix}=X_{1\tau}P_{2\psi}\ne0.
\end{equation*}
\end{proof}
Let us write out the canonical operator\footnote{Here we have
included the Maslov index in the argument of the Jacobian, write
Jacobians themselves instead of their absolute values, and
accordingly drop the factors of the form $e^{-im_j\pi/2}$
altogether.} in the coordinates $(x_1,p_2)$:
\begin{equation}\label{J-a}
[K\ph](x,h)=\BL(\frac{i}{2\pi h}\BR)^{1/2}
    \int e^{\frac ih
    S(x_1,p_2)+p_2x_2 }\frac{\ph}{\sqrt{\mathcal{J}_s}}
      dp_2.
\end{equation}
Let us apply the Fourier transform with respect to the
variable~$x_2$ to~\eqref{eq:02-10} and~\eqref{J-a} and prove that
the resulting expressions coincide modulo $O(h)$; i.e.,
\begin{equation}\label{J-c}
    \frac{1}{2\pi h}
    \int \left[e^{\frac ih
    (\tau-p_2x_2) }\ph\sqrt{\det(P,P_\psi)}
      \right]_{\tau=\tau(x,\psi)} d\psi dx_2=
    e^{\frac ih
    S(x_1,p_2)+p_2x_2 }\frac{\ph}{\sqrt{\mathcal{J}_s}}+O(h).
\end{equation}
Let us compute the expression on the left-hand side in~\eqref{J-c}
by the stationary phase method. We write
$\tau-p_2x_2=\Phi(\psi,x_2)$ and obtain, modulo $O(h)$,
\begin{equation*}
    \frac{1}{2\pi h}
    \int \left[e^{\frac ih
    (\tau-p_2x_2) }\ph\sqrt{\det(P,P_\psi)}
      \right]_{\tau=\tau(x,\psi)} d\psi dx_2\simeq
      \frac1i
      \frac{e^{\frac ih\Phi}\ph\sqrt{\det(P,P_\psi)}}
      {\sqrt{\det{(-\Phi'')}}}\bigg|_{\text{at the stationary point}},
\end{equation*}
where the argument of the determinant $\det{(-\Phi'')}$ is computed
as the sum of arguments of the eigenvalues of the $2\times2$ matrix
$-\Phi''$, the argument being take in the set $\{-\pi,0\}$.

Now let us carry out the computations at the point
$(x_1,p_2)=(X_1(\tau_0,\psi_0),P_2(\tau_0,\psi_0))$. The function
$\tau(x,\psi)$ is determined from the equation
\begin{equation}\label{J-vk}
    \langle P,x-X\rangle=0,
\end{equation}
from which, by differentiating with the properties of the eikonal
coordinates taken into account, we obtain
\begin{equation*}
    \tau_\psi=\frac{\langle x-X,P_\psi\rangle}{1-\langle
    x-X,P_\tau\rangle},\qquad
    \tau_x=\frac{P}{1-\langle
    x-X,P_\tau\rangle}.
\end{equation*}
Thus, the stationary point equations, together with \eqref{J-vk},
give
\begin{equation*}
    \langle P,x-X\rangle=0,\qquad \langle P_\psi,x-X\rangle=0,
    \qquad p_2=\frac{P_2}{1-\langle
    x-X,P_\tau\rangle},
\end{equation*}
whence, by the condition $\det(P,P_\psi)\ne0$, we have
\begin{equation*}
    x=X,\qquad p_2=P_2.
\end{equation*}
For the second derivatives of the phase function at the stationary
point, we obtain
\begin{align}
    \Phi_{\psi\psi}&=-\langle X_\psi,P_\psi\rangle,\\
    \Phi_{\psi x_2}&=P_{2\psi}-\langle X_\tau,P\psi\rangle
    P_2=P_{2\psi},\\
    \Phi_{x_2x_2}&=P_\tau P_2+P_2 \langle
    x-X,P_\tau\rangle'_{x_2} =0,
\end{align}
whence it follows that
\begin{align*}
   \det{(-\Phi'')}&=\det\begin{pmatrix}
     \langle X_\psi,P_\psi\rangle & -P_{2\psi} \\
     -P_{2\psi} & 0 \\
   \end{pmatrix}
   =-P_{2\psi}^2,\\
   \arg\det{(-\Phi'')}&\overset{\e=1}=\arg\det\begin{pmatrix}
     \e\langle X_\psi,P_\psi\rangle & -P_{2\psi} \\
     -P_{2\psi} & 0 \\
   \end{pmatrix}\\&\overset{\e\to0}=
   \arg\det\begin{pmatrix}
     0 & -P_{2\psi} \\
     -P_{2\psi} & 0 \\
   \end{pmatrix}=-\pi.
\end{align*}
(One eigenvalue is negative, and the other is positive.) Let us
show that
\begin{equation}\label{J-f}
    \frac{1}{\sqrt{\mathcal{J}_s}}=\frac{\sqrt{\det(P,P_\psi)}}
      {i\sqrt{\det{(-\Phi'')}}}.
\end{equation}
First of all, note that
\begin{equation*}
    \mathcal{J}_s \det(P,P_\psi)=
    X_{1\tau}P_{2\psi}\,P_1P_{2\psi}=P_1X_{1\tau}P_{2\psi}^2
    =(P_1X_{1\tau}+P_2X_{2\tau})P_{2\psi}^2=P_{2\psi}^2
    =-\det{(-\Phi'')}
\end{equation*}
and relation~\eqref{J-f} holds, because
\begin{equation*}
    \arg\det(P,P_\psi)=\arg\det(-\Phi'')
    +2\arg i-\arg \mathcal{J}_s=\pi-\pi-\arg \mathcal{J}_s
    =-\arg \mathcal{J}_s.
\end{equation*}
Now if we use the rule for choosing the argument of the Jacobian in
the standard singular chart and make the above-mentioned rotation,
then, in terms of the original coordinates, we obtain the rule for
choosing the argument of the singular Jacobian indicated in
Sec.~\ref{sec:02-05}.

\section{REPRESENTATION OF SOLUTIONS\\
IN A NEIGHBORHOOD OF THE CAUSTICS\\
VIA THE AIRY AND PEARCY FUNCTIONS}\label{sec:a2}

Let us show how the function~\eqref{eq:02-10} can be expressed near
the caustic via the Airy and Pearcy functions for a Lagrangian
manifold in \textit{general position} \cite{Arn2}. We include the
factors $\sqrt{\mu}$ and $e^{-i\pi m_j/2}$ in~\eqref{eq:02-10} in
the amplitude and consider the integral
\begin{equation}\label{COM}
    \cI\equiv\cI(x,h)=\Bl(\frac{i}{2\pi h}\Br)^{1/2}
    \int \bl[e^{\frac ih\tau(\psi,x)}
    A(\tau,\psi)\abs{\det(P,P_\psi)}^{1/2}\br]_{\tau=\tau(\psi,x)}d\psi,
\end{equation}
where the support $\supp A(\tau,\psi)$ lies in a neighborhood of a
focal point $(\tau^*,\psi^*)\in\La^2$. We compute the
integral~\eqref{COM} in a neighborhood of the
projection\footnote{Objects calculated at the point
$(\tau^*,\psi^*)$, are equipped with the superscript $^*$ (e.g.,
$X_{\psi\psi}^*=X_{\psi\psi}(\tau^*,\psi^*)$), which does
\textit{not} mean Hermitian conjugation in this section.}
$x^*=X^*\equiv X(\tau^*,\psi^*)$ of that point onto the
configuration space~$\mathbb{R}^2_x$. Let us subject the Lagrangian
manifold $\La^2$ to the following technical condition (which is
satisfied in all specific examples considered in Sec.~\ref{sec:03}
and without which the definitive formulas would be much more
awkward).
\begin{condition}\label{cond-t}
The Lagrangian manifold $\La^2$ lies in a level set of some
Hamiltonian of the form $H(x,p)=F(x,\abs{p})$.
\end{condition}

There exist two possible case for a focal point $(\tau^*,\psi^*)$
in general position \cite{Arn2,MaFe1}:
\begin{enumerate}
    \item[(a)] $X_{\psi\psi}^*\ne0$, or, equivalently,
$\mathcal{J}_\psi^*\ne0$ (an $A_2$ singularity, or
a~\textit{fold}).
    \item[(b)] $X_{\psi\psi}^*=0$, but $X_{\psi\psi\psi}^*\ne0$;
or, equivalently, $\mathcal{J}_\psi^*=0$, but
$\mathcal{J}_{\psi\psi}^*\ne0$ (an $A_3$ singularity, or
a~\textit{cusp}).
\end{enumerate}
(Recall that $\mathcal{J}$ is the nonsingular
Jacobian~\eqref{eq:02-Je0}.) It turns out that the
integral~\eqref{COM} can be expressed in a small neighborhood of
the point~$x^*$ of the caustic via the Airy function
\begin{equation}\label{Ai0}
\mathrm{Ai}(y)=
\frac{1}{\pi}\int_0^\infty\cos\BL(\frac{\eta^3}{3}+y\eta\BR)\,d\eta=
\frac{1}{2\pi}\int_{-\infty}^\infty\exp\BL(i\BL(\frac{\eta^3}{3}+y\eta\BR)\BR)\,d\eta
\end{equation}
in case~(a) and via the Pearcy functions
\begin{equation}\label{Pe1}
\mathrm{P}^{\pm}(v,y)= \frac{1}{2\pi}\int_{-\infty}^\infty
\exp\bigl(i(y\eta+v\eta^2\pm\eta^4)\bigr)d\eta
\end{equation}
in case~(b). Namely, the following theorem holds.
\begin{theorem}\label{L1}
Under the above-mentioned conditions, the following asymptotic
expansions hold in an $O(h^{5/6})$-neighborhood of the
point~$x^*$\rom:
\begin{multline}\label{COMAiry}
    \mathcal{I}=\frac{}
                     { }
    \frac{e^{-i\pi/4}2^{5/6}\sqrt{\pi\smash[b]{|P^*||P^*_\psi|}}}
    {\sqrt[3]{\smash[b]{|\langle P^*,X^*_{\psi\psi}\rangle|}}\sqrt[6]{h}}
    A(\tau^*,\psi^*)
    \exp \Bl(\frac{i}{h}(\tau^*+\langle P^*,x-X^*\rangle)\Br)
    \\
    \times\mathrm{Ai}\BL(-\frac{2^{1/3}\langle P^*_\psi,x-X^*\rangle }{{h^{2/3}
    \sqrt[3]{\smash[b]{\langle P^*,X^*_{\psi\psi}\rangle}} }}\BR)+O(\sqrt[6]{h})
\end{multline}
in case~\rom{(a)} and
\begin{multline}\label{COMPearcy}
\mathcal{I}=\frac{e^{-i\pi/4}\sqrt[4]{6}}
                 { \sqrt{\pi} \sqrt[4]{h}}
            \frac{\sqrt{\smash[b]{|P^*||P^*_\psi|}}}
            {\sqrt[4]{\smash[b]{|\langle
P^*_\psi,X_{\psi\psi\psi}\rangle|}}}A(\tau^*,\psi^*)
\exp{\Bl(\frac{i}{h}(\tau^*+\langle P^*,x-X^*\rangle)\Br)}
\\
\times\mathrm{P}^\pm\left(\sqrt[4]{{\frac{24}{h^3|\smash[b]{\langle
P^*_\psi,X_{\psi\psi\psi}\rangle}|}}}
  \langle P^*_\psi,x-X^*\rangle  , \sqrt{\frac{6}{h |
  \langle \smash[b]{P^*_\psi,X_{\psi\psi\psi}}\rangle|}}
  \langle P^*_{\psi\psi},x-X^*\rangle   \right)+O(1)
\end{multline}
in case~\rom{(b)}, where the upper sign on~$\mathrm{P}^\pm$ is
taken for $\langle \smash[b]{P^*_\psi,X_{\psi\psi\psi}}\rangle>0$
and the lower sign is taken in the opposite case.
\end{theorem}

\begin{proof}
First, let us derive some general formulas for integrals of rapidly
oscillating functions. Let $z=(z_1,\ldots,z_n)$ be a vector of real
parameters, $|z|\leq\varepsilon_0$, and let $\Phi(\beta,z)$ be a
smooth function with Taylor series expansion
\begin{gather*}
\Phi(\beta,z)=\Phi^{(3)}(\beta,z)+O(\beta^4)=
\Phi^{(4)}(\beta,z)+O(\beta^5),\\
\Phi^{(3)}=q_0(z)+q_1(z)\beta+\frac{q_2(z)}{2}\beta^2
+\frac{q_3(z)}{6}\beta^3+O(\beta^4),\qquad
 \Phi^{(4)}=\Phi^{(3)}+ \frac{q_4(z)}{24}\beta^4
\end{gather*}
whose coefficients~$q_j(z)$ in turn have the expansions
\begin{equation}\label{aaa}
\begin{gathered}
 q_0=a_0+\langle b_0,z\rangle,\quad q_1=a_1+\langle
 b_1,z\rangle+O(z^2), \quad q_2=a_2+\langle b_2,z\rangle+O(z^2),\\
  q_3=a_3+O(z),\quad q_4=a_4+O(z).
\end{gathered}
\end{equation}
Let $f(\beta,z)$ be a smooth function vanishing for
$\abs{\be}>\be_0$, where $\be_0$~is sufficiently small.

\begin{lemma}\label{L2}
\rom{(i)} If $a_3(0)\neq 0$, then, for $z$ in an
$O(h^{5/6})$-neighborhood of zero, one has
\begin{align}\label{Ai1}
&\int_\mathbb{R} f(\beta,z)e^{\frac{i\Phi(\beta,z)}{h}}d \beta=
\int_\mathbb{R}
f(0,z)e^{\frac{i\Phi^{(3)}(\beta,z)}{h}}d\beta+O(h^{2/3})
 \\\label{Ai2}
&\quad=2\pi f(0,z)\sqrt[3]{{\frac{2 h}
{q_3}}}\exp{\Bl(\frac{i}{h}\Bl(q_0+\frac{q_2^3 }{3 q_3^2 }-\frac{
q_1 q_2 }{ q_3 }\Br)\Br)}{\rm{sign}}(q_3)\mathrm{Ai}\Bigl( \frac{2
q_1 q_3- q_2^2 ) }{2^{2/3} q_3^{4/3} h^{2/3}}\Bigr)+O(h^{2/3})
 \\\label{Ai3}
&\quad= 2\pi f(0,0)\sqrt[3]{\frac{2 h}
{|a_3|}}\exp{\Bl(\frac{i}{h}(a_0+\langle
b_0,z\rangle)\Br)}\mathrm{Ai}\Bigl(\frac{2 \langle b_1,z\rangle
}{2^{2/3} {h^{2/3} \sqrt[3]{a_3} }}\Bigr)+O(h^{2/3}).
\end{align}
\rom{(ii)} If $a_3=0$ but $a_4\neq 0$, then, for $z$ in an
$O(h^{7/8})$-neighborhood of zero,
 \begin{align}\label{int4}
\int_\mathbb{R}
 &f(\beta,z)e^{\frac{i\Phi(\beta,z)}{h}}d \beta=
\int_\mathbb{R}
f(0,z)e^{\frac{i\Phi^{(4)}(\beta,z)}{h}}d\beta+O(h^{1/2})
 \\\label{int5}
&\begin{aligned} &=f(0,z)\sqrt[4]{\frac{24 h}
{|a_4|}}\exp{\bigl(\frac{i}{h}(q_0\pm(- q_1q_3 +\frac{ q_2 q_3^2 }{
2 q_4 }-\frac{  q_3^4 }{ 8 q_4^2
}))\bigr)}
 \\
&\qquad\qquad\times\mathrm{P}^\pm\Bigl(\sqrt[4]{\frac{24}{ h^3|q_4|
}}\bigl( q_1 + \frac{q_3^3}{3 q_4^2} - \frac{q_2q_3}{q_4}\bigr),
\sqrt{\frac{6}{h |a_4|}} (q_2 - \frac{a_3^2}{2q_4})
\Bigr)+O(h^{1/2})
\end{aligned}
 \\\label{int6}
&
\begin{aligned}
&=f(0,0)\sqrt[4]{\frac{24 h}
{|a_4|}}\exp{\bigl(\frac{i}{h}(a_0+\langle b_0,z\rangle)\bigr)}
\\&\qquad\qquad\times\mathrm{P}^\pm\Bigl(\sqrt[4]{\frac{24}{
 h^3|a_4|  }}
 \langle b_1,z\rangle , \sqrt{\frac{6}{h |a_4|}}
 \langle b_2,z\rangle  \Bigr)+O(h^{1/2}),
\end{aligned}
\end{align}
where the sign on $\mathrm{P}^\pm$ is taken according to the sign
of~$a_4$.
\end{lemma}
\begin{proof}[Proof\rm of Lemma~\ref{L2}]
(i) Since $q_3(0)\neq0$, it follows that $|q_3(z)|>C>0$ in an
$O(h^{5/6})$-neighborhood of the point $z=0$, where the
constant~$C$ is independent of $h$. This permits one to apply the
theory in \cite{Fedorjuk, ArnoldGusVarch} to the original integral
and obtain~\eqref{Ai1}. To proceed to~\eqref{Ai2}, in the
integral~\eqref{Ai0} one should make the change of variables
$\beta=q y -a_2/q_3$, $q= \sqrt[3]{{2 h}/ {q_3}}$, choosing the
sign with regard for the sign of $q_3$. Since $|z|<h^{5/6}$, we see
that the expansion~\eqref{aaa} gives
\begin{equation*}
  q_0+\frac{q_2^3 }{3 q_3^2 }-\frac{ q_1 q_2 }{ q_3 }=a_0=a_0+\langle b_0,z\rangle+O(z^2),
 \qquad
 \frac{2  q_1 q_3- q_2^2  }{2^{2/3} q_3^{4/3} h^{2/3}}=\frac{2 \langle b_1,z\rangle }{2^{2/3}
 \sqrt[3]{a_3} h^{2/3}}+\frac{O(z^2)}{h^{2/3}};
\end{equation*}
moreover, $O(z^2)h^{-2/3}=O(h)$ and $O(z^2)h^{-1}=O(h^{2/3})$.
Hence in an $O(h^{5/6})$-neighbor\-hood of the point $z=0$ one can
replace \eqref{Ai2} by \eqref{Ai3}.

The proof of~(ii) is similar. First, using the argument in
\cite{Fedorjuk,ArnoldGusVarch}, we obtain~\eqref{int4}. To obtain
\eqref{int5} for $q_4>0$, we make the change of variables $\beta=q
y - q_3/a_4$, $q= \sqrt[4]{24 h/ {q_4}}$ in the integral
$\int_\mathbb{R} f(0,z)e^{\frac{i\Phi^{(4)}(\beta,z)}{h}}d\beta$.
For $q_4<0$, we consider the complex conjugate integral and use the
substitution $q_j\to -q_j$ to reduce the proof of~\eqref{int5} to
the preceding. Again using the expansion for $q_j(z)$, we find that
\begin{gather*}
 q_0-q_1q_3+\frac{q_2 q_3^2 }{2 q_4 }-\frac{ q_3^4  }{ 8q_4^2 }
 =q_0(0)+\langle b_0,z\rangle+O(z^2),\\
 \frac{q_1+   \frac{q_3^3}{3q_4^2}- \frac{q_2 q_3}{q_4} }{ h^{3/4}\sqrt[4]{q_4}}
 =\frac{ \langle b_1,z\rangle }{ \sqrt[4]{a_4} h^{3/4}}+\frac{O(z^2)}{h^{3/4}},
\qquad
 \frac{q_2 - \frac{q_3^2}{2q_4}}{\sqrt{h q_4}}
 = \frac{\langle b_2,z\rangle }{\sqrt{h a_4}}
 +\frac{O(z^2)}{\sqrt{h}}.
\end{gather*}
If we assume that  $|z|<h^{7/8}$, then $O(z^2) h^{-3/4}=O(h)$,
$O(z^2)h^{-1/4}=O(h^{3/2})$, $O(z^2)h^{-1}=O(h^{3/4})$, and in an
$O(h^{5/6})$-neighborhood of the point $z=0$ we
obtain~\eqref{int6}.
\end{proof}

\begin{remark}
The passage from~\eqref{int4} to~\eqref{int6} produces an error of
$O(h^{3/4})$. Hence the largest error results from the truncation
of the amplitude and phase function of the original integral.
\end{remark}

Let us return to the proof of the theorem. Let us apply
Lemma~\ref{L2} to the integral~\eqref{COM}. The change of variables
$\beta=\psi-\psi^*$, $z=x-X^*$ in~\eqref{COM} gives the
integrals~\eqref{Ai1} and~\eqref{int4} with
$\Phi=\tau(\psi^*+\beta,X^*+z)$ and
$f=g(\tau,\psi)\sqrt{|\det(P,P_\psi)|}$. Next, let us compute the
coefficients~$a_j$ and $b^0,b^1,b^2$ in~\eqref{aaa} by
using~\eqref{eq:02-08} and by computing the derivatives of
$P_\psi$, $X_{\psi\psi}$ etc.\ at the focal point
$(\tau^*,\psi^*)$. By differentiating the relations $\langle
P,X_\psi\rangle=0$ and $F(x,\abs{p})=\const$ (see
condition~\ref{cond-t}) with respect to $\psi$, we obtain
 \begin{gather*}
    \langle P_\psi ,X_\psi\rangle+\langle P ,X_{\psi\psi}\rangle=0,\quad
    \langle P_{\psi\psi} ,X_\psi\rangle+2\langle P_\psi ,X_{\psi\psi}\rangle+\langle P ,X_{\psi\psi\psi}\rangle=0,\\\langle P_{\psi\psi\psi} ,X_{\psi}\rangle+3\langle P_{\psi\psi} ,X_{\psi\psi}\rangle+3\langle P_\psi ,X_{\psi\psi\psi}\rangle+\langle P ,X_{\psi\psi\psi\psi}\rangle=0,\\
    \langle P_\psi ,P\rangle=n(X)\langle n_x(X),X_\psi\rangle
     \end{gather*}
By setting  $\psi=\psi^*$ and $\tau=\tau^*$, we find that
\begin{gather}\label{eq-s2}
   X^*_\psi=0,\quad \langle P^* ,X^*_{\psi\psi}\rangle=0,\quad
    \langle P^* ,X^*_{\psi\psi\psi}\rangle=-2\langle P^*_\psi ,
    X^*_{\psi\psi}\rangle \quad\text{in case  (a)},\\X^*_\psi=X^*_{\psi\psi}=0,\quad
    \langle P^* ,X^*_{\psi\psi\psi}\rangle=0,
    \quad \langle P^* ,X^*_{\psi\psi\psi\psi}\rangle
    =-3\langle P^*_\psi ,X^*_{\psi\psi\psi}\rangle\quad\text{in case (b)},\label{eq-s3}\\
    \langle P^*_\psi ,P^*\rangle=0\label{eq-s4}
\end{gather}
Note that the 4-vector $\begin{pmatrix}
P_{\psi}\\X_{\psi}\end{pmatrix}$ is nondegenerate, because
$\dim\Lambda =2$. Hence $P^*_{\psi}\neq 0$ and
\begin{equation}\label{eq-s4a}
 |\det(P^*,P^*_{\psi})|=|P^*||P^*_{\psi}|.
\end{equation}
Let us find the coefficients~$a_j$ and~$b_j$ from the expansion
$\psi(\beta,z)=\tau(\psi^*+\beta,X*+z)\tau^*+\Delta(\beta,z)$,
where $\tau(\psi,x)$ is the solution of Eq.~\eqref{eq:02-08}; to
this end, we transform~\eqref{eq:02-08} by setting
$X^*=X(\tau^*,\psi^*)$, $X^1(\beta)=X(\psi^*+\beta,\tau^*)-X^*$,
and $\Delta=\tau-\tau^*$ and by considering the function
\begin{equation*}
  Q(\beta,\Delta)=X(\psi^*+\beta ,\tau*+\Delta)-X^*-X^1
 (\beta)- X_\tau(\psi^*+\beta ,\tau*+\Delta)\Delta.
\end{equation*}
One can readily verify that $Q|_{\Delta=0}=0$ and
$Q_\tau|_{\Delta=0}=0$ and hence $Q=O(\Delta^2)$. Let us substitute
the expansion $X(\psi^*+\beta ,\tau^*+\Delta)=X^*+X^1(\beta)+\Delta
X_\tau(\psi^*+\beta ,\tau^*+\Delta)$ into Eq.~\eqref{eq:02-08} and
take into account the fact that $\langle P,X_\tau\rangle=1$. We
obtain $\Delta= \langle P(\psi^*+\beta
,\tau^*+\Delta),x-X^*-X^1(\beta)-Q\rangle$, or $\Delta= \langle
P(\psi^*+\beta ,\tau^*+\Delta),z-X^1(\beta)\rangle+O(\Delta^2)$.
Next, $P(\psi^*+\beta ,\tau^*+\Delta)=P(\psi^*+\beta
,\tau^*)+\Delta P_\tau(\psi^*+\beta ,\tau^*)+O(\Delta^2)$ and hence
$\Delta= \langle P(\psi^*+\beta ,\tau^*)+\Delta P_\tau(\psi^*+\beta
,\tau^*),z-X^1(\beta)\rangle+O(\Delta^2)$, or
\begin{align*}
\Delta&= \frac{\langle P(\psi^*+\beta ,\tau^*),
z-X^1(\beta)\rangle} {1-\langle\frac{\pa P}{\pa \tau}(\psi^*+\beta
,\tau^*),z-X^1(\beta)\rangle}+O(\Delta^2)
\\&=
-\frac{\langle P(\psi^*+\beta
,\tau^*),X^1(\beta)\rangle}{1+\langle\frac{\pa P}{\pa
\tau}(\psi^*+\beta ,\tau^*),X^1(\beta)\rangle}+\frac{\langle
P(\psi^*+\beta ,\tau^*),z\rangle}{1+\langle \frac{\pa P}{\pa
\tau}(\psi^*+\beta ,\tau^*),X^1(\beta)\rangle}
\\
 &\qquad\qquad-\frac{\langle
P(\psi^*+\beta ,\tau^*),X^1(\beta)\rangle\langle \frac{\pa P}{\pa
\tau}(\psi^*+\beta ,\tau^*),z\rangle}{(1+\langle \frac{\pa P}{\pa
\tau}(\psi^*+\beta
,\tau^*),X^1(\beta)\rangle)^2}+O(\Delta^2)+O(z^2)
\end{align*}
By using \eqref{eq-s2} and \eqref{eq-s3}, we find that
\begin{align*}
\langle P(\psi^*+\beta ,\tau^*),X^1(\beta)\rangle
 &=\Bl\langle P^*,\frac{\beta^3}{6}X^*_{\psi\psi\psi}+
\frac{\beta^4}{24}X^*_{\psi\psi\psi\psi}\Br\rangle
 +\Bl\langle P_\psi^*,\frac{\beta^3}{2}X^*_{\psi\psi}+
\frac{\beta^4}{6}X^*_{\psi\psi\psi}\Br\rangle\\&\qquad\qquad{}+O(\beta^5)
 = \Bl\langle
P_\psi^*,\frac{\beta^3}{6}X^*_{\psi\psi}+
\frac{\beta^4}{24}X^*_{\psi\psi\psi}\Br\rangle+O(\beta^5),
\\
\Bl\langle\frac{\pa P}{\pa \tau}(\psi^*+\beta
,\tau^*),X^1(\beta)\Br\rangle&=O(\beta^2).
\end{align*}
A standard argument of the iteration method readily shows that, to
find the coefficients~\eqref{aaa}, in the last formula it suffices
to retain the terms
\begin{multline*}
-\frac{\langle P(\psi^*+\beta
,\tau^*),X^1(\beta)\rangle}{1+O(\beta^2)} +\frac{\langle
P(\psi^*+\beta ,\tau^*),z\rangle}{1+\beta^2\langle  P_\tau^*,
X_{\psi\psi}^*\rangle/2+O(\beta^3)} =-\Bl\langle
P_\psi^*,\frac{\beta^3}{6}X^*_{\psi\psi}+
\frac{\beta^4}{24}X^*_{\psi\psi\psi}\Br\rangle\\{}+O(\beta^5)+\Bl\langle
P^*+\beta
P^*_\psi+\frac{\beta^2}{2}P^*_{\psi\psi}+O(\beta^3),z\Br\rangle-\frac{\beta^2}{2}\langle
P^*,z\rangle \langle P_\tau^*,X_{\psi\psi}^*\rangle.
\end{multline*}
The iteration method gives the following formulas for the desired
coefficients:
\begin{gather}\label{coef0}
a_0=\tau^*, a_1=a_2=0, a_3=-\langle P_\psi^*,X^*_{\psi\psi}\rangle, a_4=-\langle P_\psi^*,X^*_{\psi\psi\psi}\rangle,\\
\nonumber
b_0=\langle P^*,x-X^*\rangle, b_1=\langle P^*_\psi,x-X^*\rangle, \\
b_2=\langle P^*_{\psi\psi},x-X^*\rangle-\langle P^*,x-X^*\rangle
\langle P_\tau^*,X_{\psi\psi}^*\rangle=(\text{in case (b)}=\langle
P^*_{\psi\psi},x-X^*\rangle.\label{coef1}
\end{gather}
Now, by substituting these coefficients into~\eqref{Ai3}
and~\eqref{int6} and by combining them with~\eqref{eq-s4}
and~\eqref{eq-s4a}, we arrive at the formulas in Theorem~\ref{L1}.
\end{proof}


\begin{thebibliography}{99}

\bibitem{Arnold}
V.~I.~Arnold, \textit{Funkts. Anal. i Prilozhen.}, \textbf{1}:1
(1967), 1--14. English transl: \textit{Funct. Anal. Appl.},
\textbf{1}:1 (1967), 1--13.

\bibitem{Arn1}
V.~I.~Arnold, \textit{Mathematical Methods of Classical Mechanics},
Springer-Verlag, New York, 1978.

\bibitem{Arn2}
V.~I.~Arnold, \textit{Singularities of Caustics and Wave Fronts},
Fazis, Moscow, 1996. (Russian)

\bibitem{ArnoldGusVarch}
V.~I.~Arnold, S.~M.~Gussein-Zade, and A.~N.~Varchenko,
\textit{Singularities of Differentiable Maps}, Vol.~1,
Birkha\"user, Berlin, 1985.

\bibitem{BelDob92}
V.~V.~Belov and S.~Yu.~Dobrokhotov, \textit{Teor. Mat. Fiz.},
\textbf{92}:2 (1992), 215--254. English transl: \textit{Theoret.
and Math. Phys.}, \textbf{92}:2 (1992), 843--868 (1993).

\bibitem{Berry}
M.~V. Berry and S. Klein, \textit{Proc. Natl. Acad. Sci. USA},
\textbf{93} (1996), 2614--2619.

\bibitem{arXiv}
S.~Yu.~Dobrokhotov, G.~Makrakis, and V.~E.~Nazaikinskii,
\textit{Fourier integrals and a new representation of Maslov's
canonical operator near caustics}, arXiv:1307.2292 [math-ph].

\bibitem{DobRou10}
S.~Yu.~Dobrokhotov and M.~Rouleux, \textit{Mat. Zametki},
\textbf{87}:3 (2010), 458--463. English transl: \textit{Math.
Notes}, \textbf{87}:3 (2010), 430--435.

\bibitem{DShT}
S.~Dobrokhotov, A.~Shafarevich, B.~Tirozzi, \textit{Russ. J. Math.
Phys.}, \textbf{15}:2 (2008),  192--221.

\bibitem{DShTMZ}
S.~Yu.~Dobrokhotov, B.~Tirozzi, and A.~I.~Shafarevich, \textit{Mat.
Zametki}, \textbf{82}:5 (2007),  792--796. English transl:
\textit{Math. Notes}, \textbf{82}:5 (2007), 713--717.

\bibitem{Fedorjuk}
M.~V.~Fedoryuk, \textit{Asymptotics\rom: Integrals and Series},
Nauka, Moscow, 1987.

\bibitem{Hor6}
L.~H{\"o}rmander, \textit{Acta Math.}, {\bf 127} (1971), 79--183.

\bibitem{KravtsovOrlov}
Yu.~A.~Kravtsov and Yu.~I.~Orlov, \textit{Geometric Optics of
Inhomogeneous Media}, Nauka, Moscow, 1980. (Russian)

\bibitem{Kucherenko}
V.~V.~Kucherenko, \textit{Teor. Mat. Fiz.}, \textbf{1}:3 (1969),
384--406.

\bibitem{Mas1}
V.~P.~Maslov, \textit{Perturbation theory and Asymptotic Methods},
Moscow State University, Moscow, 1965. French transl.: Dunod,
Paris, 1972.

\bibitem{Mas6}
V.~P.~Maslov, \textit{Operator Methods}, Nauka, Moscow, 1973.
English transl.: Mir, Moscow, 1976.

\bibitem{MaFe1}
V.~P.~Maslov and M.~V.~Fedoryuk, \textit{Semiclassical
Approximation for Equations of Quantum Mechanics}, Nauka, Moscow,
1976. (Russian)

\bibitem{MaNa2}
V.~P.~Maslov, V.~E.~Nazaikinskii, \textit{J. Soviet Math.},
\textbf{15}:3 (1981),  176--273.

\bibitem{MSS1}
A.~Mishchenko, V.~Shatalov, B.~Sternin, \textit{Lagrangian
Manifolds and the Maslov Operator}, Springer, Berlin, 1990.

\bibitem{Pearcy}
J.~J.~ Stamns, B.~Spjelkavik, \textit{Optica}, \textbf{30}:9
(1983), 1331--1358.

\bibitem{Vainberg}
B.~R.~Vainberg, \textit{Asymptotic Methods in Equations of
Mathematical Physics}, Moscow State University, Moscow, 1982.
(Russian)

\end{thebibliography}
\end{document}